\documentclass[aps,superscriptaddress,showpacs,pre,10pt,nofootinbib,twocolumn]{revtex4-2}

\usepackage{amsmath}
\usepackage{dsfont}
\usepackage{bm}
\usepackage{amsthm}
\usepackage{amssymb}
\usepackage[bb=boondox]{mathalfa}
\usepackage{mathtools}
\usepackage{physics}
\usepackage[colorlinks,citecolor=red,urlcolor=blue,bookmarks=true,hypertexnames=true]{hyperref} 
\usepackage[capitalize]{cleveref}
\usepackage{comment}
\usepackage{xcolor}
\usepackage[normalem]{ulem}
\usepackage{booktabs}
\usepackage[caption=false]{subfig}
\usepackage[toc,page,titletoc]{appendix}
\crefname{supp}{Supplement}{Supplements}

\newtheorem{theorem}{Theorem}
\newtheorem{lemma}{Lemma}
\newtheorem{corollary}{Corollary}
\newtheorem{proposition}{Proposition}

\renewcommand{\vec}{\bm}
\newcommand{\cvec}{\vec{c}}

\newcommand{\dz}{\mathrm{d}z}

\newcommand{\dtau}{\mathrm{d}\tau}
\newcommand{\Siq}{S_{\vec{i}_q}}

\newcommand{\avg}[1]{\langle #1 \rangle}
\newcommand{\lap}[1]{\mathcal{L}\left\{ #1 \right\}}

\newcommand{\diag}[1]{\text{diag}\left( #1 \right)}

\newcommand{\Tsteps}{N_{\mathrm{therm}}}
\newcommand{\steps}{N_{\mathrm{meas}}}
\newcommand{\stepsPer}{N_{\mathrm{skip}}}
\newcommand{\nbins}{N_{\mathrm{bins}}}
\newcommand{\nthreads}{N_{\mathrm{p}}}
\newcommand{\walltime}{T_{\mathrm{wall}}}
\newcommand{\qavg}{\avg{q}}
\newcommand{\qmax}{q_{\mathrm{max}}}
\newcommand{\avgsign}{\avg{\mathrm{sgn}}}
\newcommand{\sigsign}{\sigma_{\mathrm{sgn}}}

\newcommand*{\estimates}{\mathrel{\widehat=}}

\newcommand{\beq}{\begin{equation}}
\newcommand{\eeq}{\end{equation}}
\newcommand{\beqnn}{\begin{equation*}}
\newcommand{\eeqnn}{\end{equation*}}
\newcommand{\bea}{\begin{eqnarray}}
\newcommand{\eea}{\end{eqnarray}}
\newcommand{\beann}{\begin{eqnarray*}}
\newcommand{\eeann}{\end{eqnarray*}}
\newcommand{\bes} {\begin{subequations}}
\newcommand{\ees} {\end{subequations}}

\begin{document}

\title{A universal black-box quantum Monte Carlo approach to quantum phase transitions}

\author{Nic Ezzell} \thanks{Corresponding author: naezzell@gmail.com}
\affiliation{Department of Physics \& Astronomy, University of Southern California, Los Angeles, California 90089, USA}
\affiliation{Information Sciences Institute, University of Southern California, Marina del Rey, CA 90292, United States of America}

\author{Lev Barash}
\affiliation{Information Sciences Institute, University of Southern California, Marina del Rey, CA 90292, United States of America}
\author{Itay Hen}
\affiliation{Department of Physics \& Astronomy, University of Southern California, Los Angeles, California 90089, USA}
\affiliation{Information Sciences Institute, University of Southern California, Marina del Rey, CA 90292, United States of America}

\date{\today}

\begin{abstract}
    We derive exact, universal, closed-form quantum Monte Carlo estimators for finite-temperature 
energy susceptibility and fidelity susceptibility, applicable to essentially arbitrary Hamiltonians. 
Combined with recent advancements in Monte Carlo, our approach enables a black-box framework 
for studying quantum phase transitions — without requiring prior knowledge of an order parameter 
or the manual design of model-specific ergodic quantum Monte Carlo update rules.
We demonstrate the utility of our method by applying a single implementation to the transverse-field Ising model, 
the XXZ model, and an ensemble of models related by random unitaries.
\end{abstract}

\maketitle

\section{Introduction}
Formally, a quantum phase transition (QPT) is defined by nonanalytic behavior
of a system's ground-state energy in the thermodynamic limit~\cite{SachdevBook}.
Traditionally, QPTs are analyzed using local order parameters, such the average local magnetization 
in the transverse-field Ising model (TFIM)~\cite{SachdevBook}. 
However, identifying a suitable local order parameter for a new model is often challenging—and for topological QPTs, 
it is entirely infeasible~\cite{wen2017ColloquiumZooQuantumtopological, hamma2008EntanglementFidelityTopological}.
Consequently, substantial effort has been devoted to developing generic methods for studying QPTs that do not rely 
on prior knowledge of an order parameter~\cite{FIDELITYAPPROACHQUANTUM, zanardi2006GroundStateOverlap, zanardi2007BuresMetricThermal, zanardi2007InformationTheoreticDifferentialGeometrya, zanardi2007MixedstateFidelityQuantum, you2007FidelityDynamicStructure, rigol2009FidelitySuperconductivityTwodimensional,kasatkin2024detecting, sun2015FidelityBerezinskiiKosterlitzThoulessQuantum,kasatkin2024detecting,zhou2008GroundStateFidelity,jiang2011ReducedFidelityEntanglement,jordan2009NumericalStudyHardcore,wang2010BerezinskiiKosterlitzThoulessTransitionUncovered, schwandt2009QuantumMonteCarlo, albuquerque2010QuantumCriticalScaling, wang2015FidelitySusceptibilityMade, damski2013FidelitySusceptibilityQuantum, bialonczyk2021UhlmannFidelityFidelity, sirker2010FiniteTemperatureFidelitySusceptibility}.
Nevertheless, these approaches still rely on manually encoding subtle, system-specific prior knowledge.
Building on recent advancements in quantum Monte Carlo (QMC) methods~\cite{barash2024QuantumMonteCarlo, akaturk2024quantum, babakhani2025quantum, ezzell2025advanced} 
and the well-established fidelity-based approach to QPTs~\cite{FIDELITYAPPROACHQUANTUM, schwandt2009QuantumMonteCarlo, albuquerque2010QuantumCriticalScaling, sirker2010FiniteTemperatureFidelitySusceptibility, wang2015FidelitySusceptibilityMade}, 
we propose an automated, order-parameter-free, black-box method for studying QPTs that requires no manual encoding of system-specific details.

More precisely, we derive QMC estimators for \emph{energy susceptibility} (ES)~\cite{albuquerque2010QuantumCriticalScaling, chen2008IntrinsicRelationGroundstate, camposvenuti2007QuantumCriticalScaling} 
and \emph{fidelity susceptibility} (FS)~\cite{FIDELITYAPPROACHQUANTUM, you2007FidelityDynamicStructure, schwandt2009QuantumMonteCarlo, albuquerque2010QuantumCriticalScaling, wang2015FidelitySusceptibilityMade,tzeng2008fidelity}, 
which can identify quantum critical points in QPTs without relying on order parameters. 
While related estimators have appeared in other QMC frameworks~\cite{albuquerque2010QuantumCriticalScaling, schwandt2009QuantumMonteCarlo, wang2015FidelitySusceptibilityMade}, 
our derivations and resulting estimators are developed within the relatively 
recent permutation matrix representation QMC (PMR-QMC), which offers several key advantages 
over earlier methods~\cite{gupta2020PermutationMatrixRepresentation, hen2018OffdiagonalSeriesExpansion, albash2017OffdiagonalExpansionQuantum, barash2024QuantumMonteCarlo, akaturk2024quantum, babakhani2025quantum, ezzell2025advanced}. 
Most relevant to this work, it was recently shown that ergodic PMR-QMC update rules satisfying detailed balance 
can be generated via a deterministic, automated algorithm for essentially arbitrary 
Hamiltonians~\cite{barash2024QuantumMonteCarlo, akaturk2024quantum, babakhani2025quantum}. 
Furthermore, our estimators are applicable to arbitrarily complex driving terms, as we demonstrate in this work.
To the best of our knowledge, neither of these capabilities was supported by previous approaches. 
Taken together, these features allow our method to be applied directly to a wide range of systems 
without the need to manually encode laborious, system-specific details.

To this end, we apply our approach to estimate ES and FS in three spin systems: a 2D TFIM model, 
a 2D XXZ model, and a 100-spin model with random Pauli operator terms that—to our knowledge—cannot be studied using any other method. 
Given any spin-$1/2$ Hamiltonian expressed as a sum of Pauli operator products, our current implementation~\cite{ezzell2025code} 
can attempt to locate quantum critical points by estimating ES and/or FS. 
As with any QMC method, our approach may encounter convergence issues due to frustration 
or the sign problem~\cite{hen2021determining, pan2024SignProblemQuantum}, 
which can sometimes be mitigated with custom QMC updates~\cite{kandel1990cluster, rakala2017cluster} 
or other techniques~\cite{gupta2020elucidating, hen_resolution_2019, hen2021determining, melko2013stochastic, sandvik2019stochastic, li2019SignProblemFreeFermionicQuantum, pan2024SignProblemQuantum}. 
For the three models studied in this work, however, no manual, system-specific adjustments were required 
to ensure convergence. In all cases, our code automatically monitors the severity 
of the sign problem~\cite{barash2024PmrQmcCode, barash2024QuantumMonteCarlo} 
and performs thermalization and autocorrelation diagnostics to ensure reliable estimates.

Our implementation~\cite{ezzell2025code} builds upon preexisting open-source code for arbitrary spin-$1/2$ models~\cite{barash2024PmrQmcCode}. 
Since our estimators are derived within the general PMR-QMC framework, which can represent arbitrary Hamiltonians~\cite{gupta2020PermutationMatrixRepresentation, ezzell2025advanced}, 
they can readily be ported to any existing or future PMR-QMC codebases. 
At present, this includes Bose-Hubbard models on arbitrary graphs~\cite{akaturk2024quantum}, 
arbitrary spin-$1/2$, high-, and mixed-spin systems~\cite{barash2024QuantumMonteCarlo,babakhani2025quantum}, and fermionic systems~\cite{babakhani2025quantum,babakhani2025fermionic}. 
In summary, PMR-QMC can automatically generate QMC moves for essentially any Hamiltonian~\cite{barash2024QuantumMonteCarlo, akaturk2024quantum, babakhani2025quantum}, 
and in this work, we derive generic ES and FS estimators within this framework. 
Combined with automatic convergence diagnostics, our approach—unlike previous numerical methods—can reliably be used as a 
black-box tool for studying QPTs.

\section{Results}
\label{sec:numerics}
Before deriving formal estimators in \cref{sec:methods}, we demonstrate our method through a series of numerical experiments 
using open source PMR-QMC code~\cite{ezzell2025code} that builds upon prior work~\cite{barash2024PmrQmcCode,barash2024QuantumMonteCarlo} 
(see \cref{app:pmr-qmc-background} for basics of PMR-QMC and \cref{app:code-and-empirical-resource-scaling}  for basic code instructions). 
Each model is a spin-$1/2$ system defined in terms of the standard Pauli operators $X$, $Y$, and $Z$. 

We use the notation for the ES $\chi_E$ and the FS $\chi_F$ as defined in ``Theoretical background: Energy and fidelity susceptibilities" in \cref{sec:methods}; 
the PMR-QMC conventions follow ``PMR-QMC background and notation" also in \cref{sec:methods}.
$X_i$, $Y_i$, and $Z_i$ denote Pauli $X$, $Y$, and $Z$ operators acting on spin $i$, respectively.
{ All error bars and continuous error bands estimate the spread of the plotted quantity over independent QMC runs: Either $2 \sigma$ (for thermal averages) or direct 95\% percentile intervals (for wall-clock times). }

We begin with illustrative experiments on well-known models, emphasizing agreement with expected results and, 
where applicable, reproducing prior plots~\cite{albuquerque2010QuantumCriticalScaling, wang2015FidelitySusceptibilityMade}. 
We then present a more challenging experiment involving an ensemble of random models, for which our approach 
still accurately estimates ES and FS. Notably, for this latter case, we are unaware of any other method 
capable of studying the model successfully. 
Yet, our single code can handle all of these models with only straightforward adjustments to input files~\cite{ezzell2025code}. \\

\noindent \textbf{A simple two spin model}

\begin{figure}
   \centering
    \includegraphics[width=0.45\textwidth]{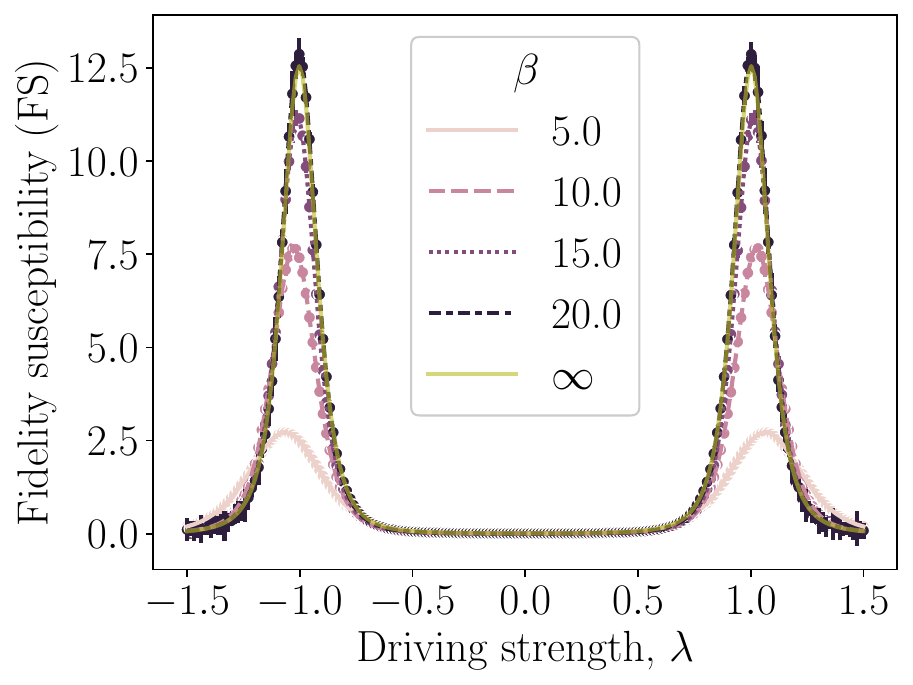}
    \caption{FS for the two spin model, Eq.~(\ref{eq:prl-model}), considered in Ref.~\cite{zhang2008detection} for different inverse temperatures, $\beta$. Data from QMC are averaged over 5 independent random seeds, whereas direct numerical calculations are shown as lines.}
    \label{fig:prl-comparison-1}
\end{figure}

We begin by studying,
\begin{equation}
    \label{eq:prl-model}
    H(\lambda) = Z_1 Z_2 + \gamma (X_1 + X_2) + \lambda(Z_1 + Z_2), 
\end{equation}
first introduced and studied in a quantum sensing experiment on an NMR system via the fidelity approach (we choose $\gamma=0.1$)~\cite{zhang2008detection}. By design, this model exhibits an avoided level crossing near the critical points $\lambda \approx \pm 1$, where the ES and FS are expected to peak. Our QMC approach correctly locates these peaks by computing the FS as shown in~\cref{fig:prl-comparison-1}. In fact, as shown, our approach correctly matches direct numerical calculations for each temperature tested. Furthermore, our plot shows expected convergence of $\chi_F^\beta \rightarrow \chi_F$ as $\beta \rightarrow \infty$. At $\beta = 20$ specifically, the QMC results closely match the exactly computed $T = 0$ results. This verifies the basic veracity of our method.  

Next, we remark that prior QMC approaches are generally incapable of estimating the FS, as they must typically restrict themselves to models for which $H_1 \propto \diag{H}$ or $H_1 \propto H - \diag{H}$~\cite{schwandt2009QuantumMonteCarlo,  albuquerque2010QuantumCriticalScaling,wang2015FidelitySusceptibilityMade,sirker2010FiniteTemperatureFidelitySusceptibility}.
Despite the simplicity of this model,  it does not satisfy this condition.  Yet,  our method easily handles this model without even needing to write a specialized code for this case. \\

\begin{figure}
    \includegraphics[width=0.45\textwidth]{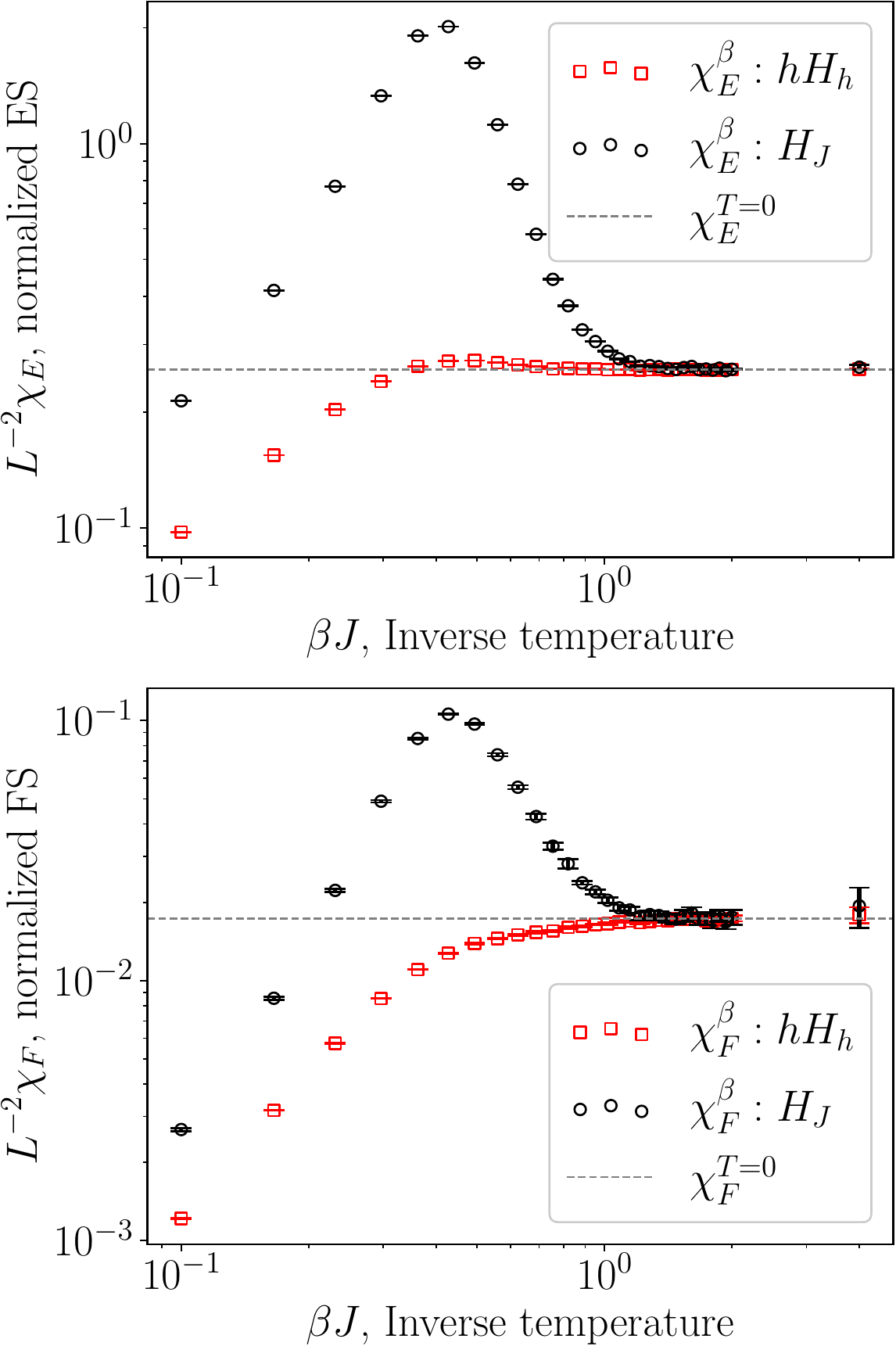}
    \caption{We compute ES and FS restricted to the positive parity subspace (see \cref{app:tfim-parity-details}) of a $4 \times 4$ square TFIM with PBC and $h = 1$ as a function of $\beta$. {  QMC data are averaged over 100 independent runs.} Our FS results are consistent with a similar computation performed with SSE-QMC~\cite{albuquerque2010QuantumCriticalScaling}.}
    \label{fig:prb-replicate-fig3}
\end{figure}

\noindent \textbf{A square lattice TFIM}

Secondly,  we consider the well-known and extensively studied~\cite{albuquerque2010QuantumCriticalScaling, blote1995ising, blote2002cluster} transverse-field Ising model (TFIM), 
\begin{equation}
    \label{eq:fidsus-tfim}
    H_{\text{TFIM}}(h) = H_J + h H_h = - \sum_{\avg{i, j}} X_i X_j - h \sum_{i=1}^n Z_i,
\end{equation}
on a square lattice with periodic boundary conditions (PBC).   For our purposes, we focus our efforts on a simple $4 \times 4$ model at fixed $h = 1$ and  compute the ES and FS as a function of inverse temperature for either $H_J$ or $h H_h$ as the ``driving term.'' Notably, these results are actually not unconditional ES and FS as defined in~\cref{eq:gibbs-ES,eq:gibbs-FS}, but rather, the result of restricting computations only to the positive parity subspace of the TFIM.  We refer interested readers to \cref{app:tfim-parity-details} for details on this methodological trick, but we remark that our approach is a non-trivial convenience of PMR-QMC that highlights its versatility.  

\begin{figure}
    \centering
       \includegraphics[width=0.45\textwidth]{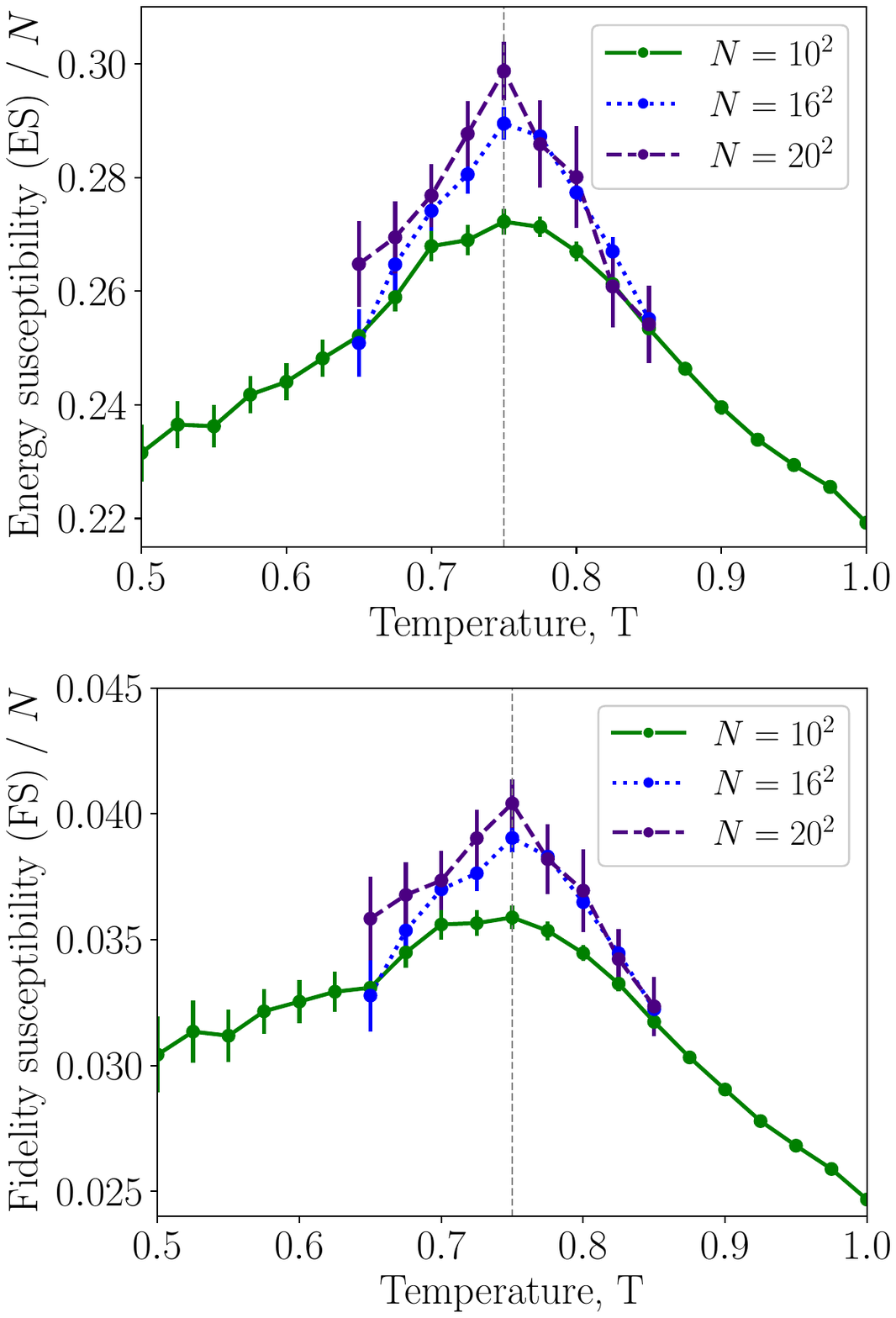}
    \caption{We compute ES and FS for the XXZ model in~\cref{eq:fidsus-xxz} on a square lattice with PBC as a function of temperature, $T$, for different model sizes. { QMC data are averaged over at least 200 independent runs.} Our FS results are consistent with a similar computation performed with SSE-QMC~\cite{wang2015FidelitySusceptibilityMade}.}
    \label{fig:prx-replicate-fig5}
\end{figure}

The results are shown in~\cref{fig:prb-replicate-fig3} and demonstrate three things.  Firstly, they showcase the ability of our method to consider different driving term choices as well as compute things in a fixed parity subspace.  Secondly, they help verify the correctness of our approach, as the FS plot correctly matches results independently calculated with SSE-QMC~\cite{albuquerque2010QuantumCriticalScaling}. Finally,  they correctly reveal an expected feature of the ES and FS.  Namely, one might naively expect that, for the purposes of detecting a quantum critical point, one is free to view either the static term $H_J$ or the driving term $h H_h$ in~\cref{eq:fidsus-tfim} as the perturbation.

In the $T = 0$ limit, this is correct since 
\begin{equation}
    \chi_{F,0}^{H_J} = h^2 \chi_{F,0}^{h H_h}, \ \ \ \chi_{E,0}^{H_J} = h^2 \chi_{E,0}^{h H_h},
\end{equation}
follows from plugging in  $H_J = H - h H_h$ into~\cref{eq:spectral-ES,eq:spectral-FS} alongside the observation $\avg{\psi_n | H | \psi_0} = 0$ for all $n \neq 0$. In these plots, we have chosen $h = 1$, and correspondingly, we find that both the ES and FS converge for either perturbation choice ($H_J$ or $h H_h)$ when $\beta$ is large enough. Yet for any finite $\beta$, the relation is less trivial, and in fact, one finds that there is no simple relation that holds across all inverse temperatures~\cite{albuquerque2010QuantumCriticalScaling}. Methodologically, this provides a different way to test for empirical convergence of the finite temperature ES and FS to groundstate values.

So far,  we have considered simple models of fixed size.  To connect these ideas to the study of QPTs in the infinite system limit,  we remark that one can simply compute the ES and FS for increasingly larger system sizes.  This is done, for example, in Fig. 4 of Ref.~\cite{albuquerque2010QuantumCriticalScaling}. As expected, this results in an increasingly sharper peak in both the ES and FS as $h$ approaches the known critical point $h_c \approx 3.044$~\cite{blote2002cluster}.  In addition, one can perform finite size scaling and data collapse in accordance with expected scaling relations,~\cref{eq:es-scaling-relation,eq:fs-scaling-relation}, to extract the critical exponent $\nu \approx 0.6301$~\cite{blote1995ising, albuquerque2010QuantumCriticalScaling}.  Expressed differently,  our method is straightforwardly compatible with studying QPTs in the infinite size limit with enough computational effort,  which we show more directly for the XXZ model next.  \\

\noindent \textbf{A square lattice XXZ model}

\begin{figure*}
\includegraphics[width=\textwidth]{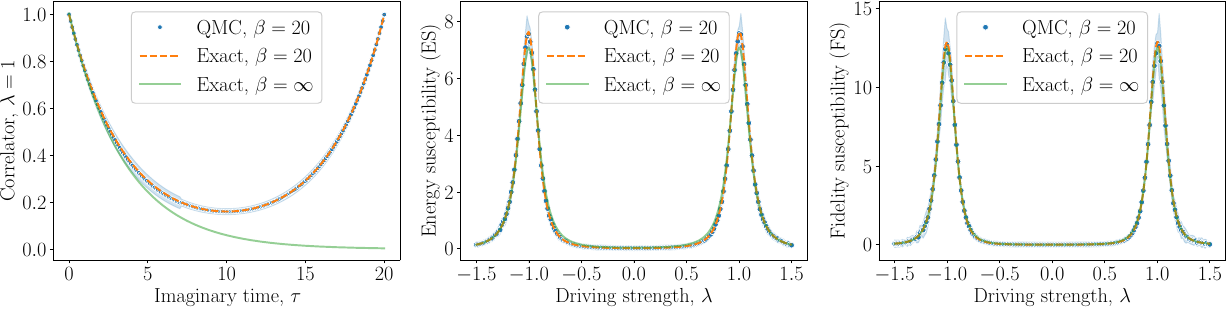}
\caption{Comparison of QMC estimates with exact values for an ensemble of 100--spin Hamiltonians generated 
by random rotations of Eq.~\eqref{eq:prl-model}. Detailed information on the ensemble can be found 
in \cref{app:random-U}.  QMC estimates are averaged across 10 random models. In all cases, 
the QMC estimates show excellent agreement with exact finite $\beta$ curves.
(a) The imaginary time correlator, $G(\tau)$, has the expected symmetry around $\beta/2 = 10.$
(b) The ES correctly predicts critical points despite the $\beta = 20$ curve not quite converging to the $\beta = \infty$ curve.
(c) The FS correctly predicts critical points more sharply than ES, and the $\beta = 20$ curve matches the $\beta = \infty$ curve.
}
\label{fig:prl-rot-obs}
\end{figure*}

Thirdly,  we study the antiferromagnetic XXZ model on a square lattice with PBC, 
\begin{equation}
    \label{eq:fidsus-xxz}
    H = \frac{J_z}{4} \sum_{\avg{i,j}} Z_i Z_j + \frac{\lambda}{4} \sum_{\avg{i,j}} (X_i X_j + Y_i Y_j),
\end{equation}
which relates to a hard-core boson model~\cite{schmid2002finite,koziol2024quantum} via the Matsubara-Matsuda transformation~\cite{matsubara1956lattice}. The finite-temperature phase diagram of this model has been studied intensely~\cite{schmid2002finite}, and various facts of this diagram have been replicated by FS approaches with exact diagonalization~\cite{yu2009fidelity} and SSE-QMC~\cite{wang2015FidelitySusceptibilityMade}.  From a QPT perspective, this model prefers  N{\'e}el order in the XY plane when $\lambda$ dominates and an antiferromagnetic Ising groundstate when $J_z$ dominates. The point $\lambda = J_z$ is a known quantum critical point known as the Heisenberg point~\cite{wang2015FidelitySusceptibilityMade, schmid2002finite}.

Furthermore~\cite{wang2015FidelitySusceptibilityMade, schmid2002finite}, it is known that thermal fluctuations at finite temperature destroy the antiferromagnetic Ising phase at a second-order phase transitions. Fixing $\lambda = 1, J_z = 1.5$ and varying temperature $T$, one can search for the expected critical point $(T / \lambda)_c \approx 0.75$~\cite{wang2015FidelitySusceptibilityMade, schmid2002finite}. Using the ES and FS, we search for this critical point in~\cref{fig:prx-replicate-fig5}, and the empirical peak as we scale system size agrees with expectations.  The FS plot in particular is consistent with results obtained in Ref.~\cite{wang2015FidelitySusceptibilityMade}. More broadly,  these experiments show that our method can scale to large system sizes (which are necessary for finite size scaling in the study of QPTs) and also can study finite temperature phase transitions. \\

\noindent \textbf{A random ensemble}

Finally,  we consider an ensemble of random 100-spin Hamiltonians generated by random unitary rotations of Eq.~(\ref{eq:prl-model}) (see \cref{app:random-U} for full details of this family).  By unitary invariance,  observables for these random models agree with the simple two qubit model despite each $H_{U_i} = U_i  H(\lambda) U_i^{\dagger}$ and associated driving term $U_i H_1 U^{\dagger}_i$ being quite complicated,  containing hundreds of random Pauli terms over 100 spins.  Using standard QMC approaches, it is not possible to study this ensemble,  but it is straightforward to try in PMR-QMC,  and the results agree with known expectations as shown in~\cref{fig:prl-rot-obs}.   

Specifically,  PMR-QMC estimators consistently show excellent agreement with exact numerical results for  the imaginary-time correlator $G(\tau)$, the ES, and the FS, as shown in Fig.~\ref{fig:prl-rot-obs}.  Furthermore, our results confirm established characteristics of these observables; for instance, $G(\tau)$ exhibits symmetry about $\beta/2$ (see Fig.~1 in Ref.~\cite{albuquerque2010QuantumCriticalScaling} and Fig.~7 in Ref.~\cite{wang2015FidelitySusceptibilityMade}). \\

{ 
\noindent \textbf{Simulation resources and scaling}
\begin{figure}
    \centering
    \includegraphics[width=0.4\textwidth]{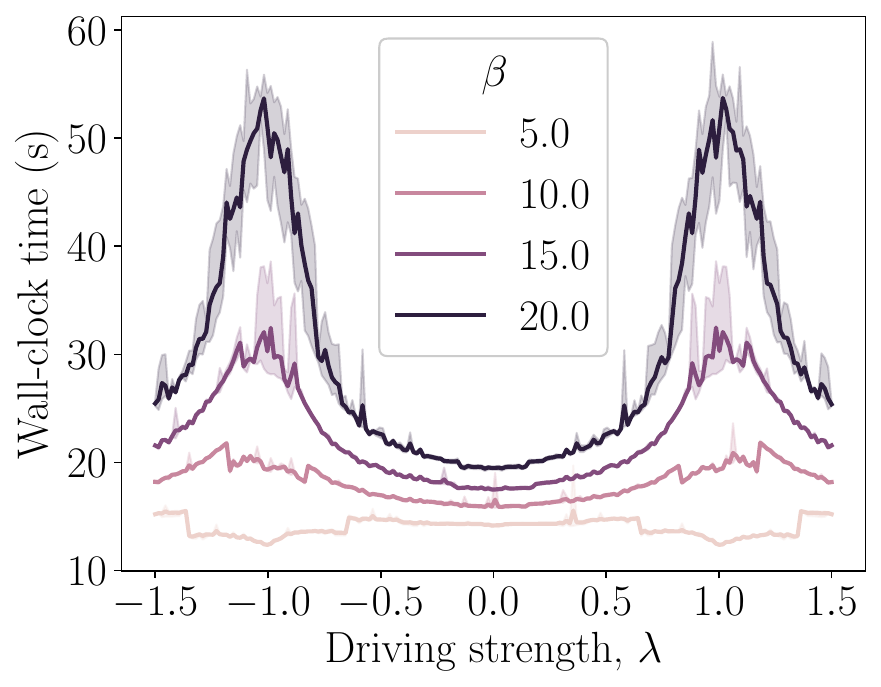}
    \caption{ Simulation time used to generate \cref{fig:prl-comparison-1}.}
    \label{fig:maintext-timing-1}
\end{figure}

\begin{table}
\centering
\begin{tabular}{lccc}
\toprule
Fig. & $\nthreads$  & $\min(\walltime)$ & $\max(\walltime)$ \\
~\ref{fig:prl-comparison-1} & 5 & 10s & 1m \\
~\ref{fig:prb-replicate-fig3} & $100$ & 15s & 14.4hrs \\
~\ref{fig:prx-replicate-fig5} & $200-400$ & 3.6hrs & 8.8days \\
~\ref{fig:prl-rot-obs} & $5$ &  4s & 8.2m \\
\midrule
\bottomrule
\end{tabular}
\caption{ We summarize the resources used to generate data in this work. Further information is contained in \cref{app:code-and-empirical-resource-scaling}.}
\label{tab:simulation_parameters}
\end{table}

\begin{figure*}
    \centering
    \includegraphics[width=0.9\textwidth]{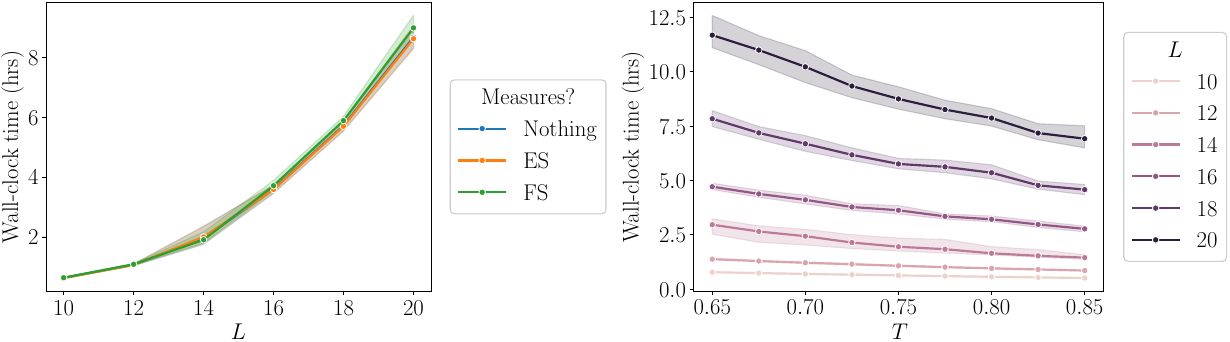}
    \caption{ Estimate of total simulation time scaling for the $L \times L$ XXZ model data in \cref{fig:prx-replicate-fig5}. Timings are given with the following simulation parameters fixed: $\Tsteps = 0$, $\steps = 4\times 10^8$, $\stepsPer = 5\times 10^3$. Errors bars are estimated over 20 independent runs ($\nthreads = 20)$. We refer readers to \cref{app:code-and-empirical-resource-scaling} for more details on the meaning of simulation parameters. (Left) Wall-clock time as a function of $L$ at the critical temperature $T_c \approx 0.75$, separated by what is being measured.  Simulation time scales roughly as $L^4 = N^2$ in this case. (Right) Wall-clock times to simulate the XXZ model as a function of $T$ for different $L$, with all measurement approaches averaging together.}
    \label{fig:maintext-timing-2}
\end{figure*}

Numerical demonstrations for \cref{fig:prl-comparison-1,fig:prl-rot-obs}, corresponding to studies of \cref{eq:prl-model} and its random rotations, were performed on a single 2.3\,GHz Intel i9 laptop processor. The remaining experiments involving the TFIM and XXZ models (\cref{fig:prb-replicate-fig3,fig:prx-replicate-fig5}) were executed on a University of Southern California high-performance computing cluster. The algorithm is highly memory-efficient~\cite{barash2024QuantumMonteCarlo,gupta2020CalculatingDividedDifferences}, with modest requirements even for large-scale simulations (approximately 200\,MB of memory per run was sufficient for all calculations reported here). Consequently, our performance analysis focuses primarily on time.

As a QMC method, our approach is naturally suited for lazy parallel execution
and exhibits near-perfect strong scaling as the number of cores increases
(see, e.g., Fig.~8 of Ref.~\cite{barash2024QuantumMonteCarlo}).
\Cref{tab:simulation_parameters} summarizes the number of MPI threads ($\nthreads$, approximately corresponding to the number of cores) and wall-clock times used to produce the results in \cref{fig:prl-comparison-1,fig:prb-replicate-fig3,fig:prx-replicate-fig5,fig:prl-rot-obs}. Additional details regarding simulation parameters and computational resources are provided in \cref{app:code-and-empirical-resource-scaling}.

Next, we briefly discuss the empirical scaling of our algorithm.
In \cref{fig:maintext-timing-1}, we plot the total simulation time used to generate \cref{fig:prl-comparison-1} in our study of the two qubit model in \cref{eq:prl-model}. Though a simple model, simulation time grows as a function of $\beta$ and near the avoided level crossings, as expected.

In \cref{fig:maintext-timing-2}, we repeat our study of the XXZ model at the critical temperature $T_c \approx 0.75$ at (a) different system sizes $N = L^2$ and (b) with different measurements protocols. We also show the modest wall-clock time scaling as a function of $T$ for different $L$.   In our timing experiments, all simulations are run on the same 2.6 GHz AMD EPYC 7513 CPU and simulation parameters are fixed to values listed in the caption, chosen to be in the same range as those used for \cref{fig:prx-replicate-fig5}. The different measurements correspond to either running QMC without performing any measurement, while measuring only the ES, or while measuring only the FS. In all cases, times are similar and the scaling is roughly proportional to $N^2 = L^4$. Given that ES and FS have $O(1)$ and $O(q^3)$ estimators (see \cref{thrm:constant-ES-estimator,thrm:cubic-FS-estimator}), this may be surprising, but measurements are performed sparsely (every $\stepsPer = 5000$ MC moves), so wall-clock time is dominated by QMC simulation itself.  Further discussion and empirical testing of resource scaling is given in \cref{app:code-and-empirical-resource-scaling}. 

\section{Discussion}

We have developed a universal black-box, finite-temperature quantum Monte Carlo (QMC) approach 
for studying quantum phase transitions (QPTs).
Specifically, we derived PMR-QMC estimators for energy susceptibility (ES)
and fidelity susceptibility (FS), which serve as general-purpose, order-parameter-free tools for locating quantum critical points. 
At the same time, given a Hamiltonian, we can compute a set of ergodic QMC updates that satisfy detailed balance.

Thanks to the generality of both the PMR-QMC framework and our PMR-based derivations, we successfully studied 
various aspects of QPTs, ES, and FS for a transverse-field Ising model, an XXZ model, and even a random ensemble—all using a single code. 
These demonstrations required no model-specific modifications to the source code—only a specification of the Hamiltonian 
and the desired partition, both of which can be arbitrary.  
As such, our method and publicly available open-source code are designed to be accessible and user friendly.

Of note, our method can estimate the ES in constant time for models with a simple driving term such as in the TFIM or XXZ model. Though the ES is expected to be a weaker indicator of criticality than the FS in general, we have found it to be sufficient in all models studied in this work. Hence, our approach is not only general but also very efficient for many systems of interest in condensed matter physics. 

There are several promising avenues for future research. Foremost, we aim to apply our approach to physically relevant models 
with non-standard driving terms. We also intend to extend our method to bosonic, fermionic, and high-spin systems of current interest.

We hope the results of this work will provide a valuable and versatile tool for condensed matter physicists 
and researchers studying quantum phase transitions.

\section{Methods}
\label{sec:methods}

We introduce requisite technical background then state and discuss the novel PMR-QMC estimators utilized in our study of QPTs in \cref{sec:numerics}.  Importantly,  we state  \cref{thrm:constant-ES-estimator,thrm:cubic-FS-estimator},  which provide estimators for ES and FS applicable to most well-studied condensed matter models.  At the cost of some efficiency,  we also generalize these theorems to arbitrary Hamiltonians.  \\

\noindent \textbf{Energy and fidelity susceptibilities}

We discuss the ES and the FS as two generic, order-parameter-free tools to study QPTs~\cite{albuquerque2010QuantumCriticalScaling, wang2015FidelitySusceptibilityMade, FIDELITYAPPROACHQUANTUM, you2007FidelityDynamicStructure}. Throughout, we assume an $n$ spin Hamiltonian with non-degenerate groundstate,
\begin{equation}
    \label{eq:qpt-ham}
    H(\lambda) = H_0 + \lambda H_1,
\end{equation}
that has \emph{quantum critical point} (QCP) at $ \lambda = \lambda_c,$ denoting the boundary between two quantum phases of a QPT. Henceforth, we refer to $H_0$ as the static term and $H_1$ as the driving term for convenience. Given the spectral decomposition, $H(\lambda) \ket{\psi_k(\lambda)} = E_k(\lambda) \ket{\psi_k(\lambda)},$ we define the ES~\cite{albuquerque2010QuantumCriticalScaling, chen2008IntrinsicRelationGroundstate, camposvenuti2007QuantumCriticalScaling}
\begin{equation}
    \label{eq:2nd-der-ES}
    \chi_E(\lambda) = - \frac{\partial^2 E_0(\lambda)}{\partial\lambda^2}.
\end{equation}
For second-order QPTs such as that in the TFIM~\cite{SachdevBook}, the ES is expected to have a peak or even diverge, and hence, it can be used to locate QCPs. Similarly, the second-order term in the Taylor expansion of groundstate fidelity, or FS~\cite{you2007FidelityDynamicStructure}, 
\begin{equation}
    \label{eq:taylor-FS}
    \chi_F(\lambda) = \lim_{\delta \lambda \rightarrow 0} \frac{-2 \ln | \avg{\psi_0(\lambda) | \psi_0(\lambda + \delta \lambda)} |}{\delta \lambda^2},
\end{equation}
also exhibits a peak or diverges at QCPs for a wide variety of QPTs~\cite{camposvenuti2007QuantumCriticalScaling, gu2008FidelitySusceptibilityScaling, albuquerque2010QuantumCriticalScaling}, including first-order~\cite{kasatkin2024detecting}, second-order~\cite{kasatkin2024detecting}, topological~\cite{abasto2008FidelityAnalysisTopological, yang2007GroundstateFidelityOnedimensional}, and as seen recently, even dynamical QPTs~\cite{kasatkin2024detecting}.

There are several equivalent characterizations of both quantities that are useful in different contexts. First, we discuss various rich characterizations of the FS~\cite{FIDELITYAPPROACHQUANTUM, you2007FidelityDynamicStructure, gu2008FidelitySusceptibilityScaling, camposvenuti2007QuantumCriticalScaling, zanardi2007InformationTheoreticDifferentialGeometrya}. By a direct Taylor expansion and simplification of $\avg{\psi_0(\lambda) | \psi_0(\lambda + \delta \lambda)},$  one finds the differential form,
\begin{multline}
    \chi_F(\lambda) = \avg{\partial_\lambda \psi_0(\lambda) | \partial_\lambda \psi_0(\lambda)} \\
    - \avg{\partial_\lambda  \psi_0(\lambda)| \psi_0(\lambda)} \avg{  \psi_0(\lambda)| \partial_\lambda \psi_0(\lambda)},
\end{multline}
which has a clear relation to the quantum metric tensor, i.e., the real part of the quantum geometric tensor~\cite{camposvenuti2007QuantumCriticalScaling, zanardi2007InformationTheoreticDifferentialGeometrya}. Using non-degenerate perturbation theory, we can also derive the spectral form, 
\begin{align}
    \label{eq:spectral-FS}
    \chi^{H_1}_F(\lambda) &= \sum_{n \neq 0} \frac{ | \avg{\psi_n(\lambda) | H_1 | \psi_0(\lambda)} |^2 }{[E_n(\lambda) - E_0(\lambda)]^2},
\end{align}
where the superscript now indicates which term we treat as driving the perturbation. Finally, we can derive an integral representation~\cite{you2007FidelityDynamicStructure} that is useful in deriving both scaling relations~\cite{gu2008FidelitySusceptibilityScaling, albuquerque2010QuantumCriticalScaling} and in generalizing the Gibbs state for QMC~\cite{albuquerque2010QuantumCriticalScaling, schwandt2009QuantumMonteCarlo, wang2015FidelitySusceptibilityMade}. 

To proceed~\cite{albuquerque2010QuantumCriticalScaling}, we define the imaginary time evolved operator $H_1(\tau) \equiv e^{\tau H(\lambda)} H_1 e^{-\tau H(\lambda)}$ and the connected correlation function,
\begin{multline}
    \label{eq:gs-correlator}
    G^{H_1}_0(\tau ; \lambda) \equiv \theta(\tau) ( \avg{ \psi_0(\lambda) | H_1(\tau) H_1 | \psi_0(\lambda)} \\ 
    - \avg{\psi_0(\lambda) | H_1 | \psi_0(\lambda)}^2),
\end{multline}
for $\theta(\tau)$ the Heaviside step function. Dropping $\lambda$ dependence for simplicity, inserting the resolution of the identity between $H_1$ and $e^{-\tau H}$ and simplifying, we find
\begin{equation}
    G^{H_1}_0(\tau ; \lambda) = \theta(\tau) \sum_{n \neq 0} e^{\tau (E_0 - E_n) } | \avg{\psi_0 | H_1 | \psi_n}|^2.
\end{equation}
Taking the inverse Fourier transform,
\begin{equation}
    \label{eq:fourier-correlator}
    \widetilde{G}_0^{H_1}(\omega) \equiv \int_{\mathds{R}} e^{-i \tau \omega} G^{H_1}_0(\tau) = \sum_{n \neq 0} \frac{ | \avg{\psi_0 | H_1 | \psi_n}|^2  }{ E_n - E_0 + i \omega },
\end{equation}
we see a strong resemblance between this formula and~\cref{eq:spectral-FS}. Namely, we recognize, 
\begin{equation}
    \label{eq:gs-int-FS}
    \chi^{H_1}_{F}(\lambda) = i \frac{\mathrm{d} \widetilde{G}_0^{H_1}(\omega)}{\mathrm{d} \omega} = \int_0^\infty \tau  G^{H_1}_0(\tau ; \lambda) \dtau,
\end{equation}
which is the sought after integral representation. This form relates $\chi_F$ to the dynamical response of the system to $H_1$, revealing its physical content~\cite{you2007FidelityDynamicStructure, gu2008FidelitySusceptibilityScaling, albuquerque2010QuantumCriticalScaling}. 

A similar story holds for ES~\cite{albuquerque2010QuantumCriticalScaling}. Applying non-degenerate perturbation theory, we find the second derivative of energy is given by
\begin{align}
    \label{eq:spectral-ES}
    \chi^{H_1}_E(\lambda) &= 2 \sum_{n \neq 0} \frac{ | \avg{\psi_n(\lambda) | H_1 | \psi_0(\lambda)} |^2 }{E_n(\lambda) - E_0(\lambda)},
\end{align}
which---aside from the pre-factor of $2$---differs from~\cref{eq:spectral-FS} only by the denominator. By inspection of~\cref{eq:spectral-ES} and~\cref{eq:fourier-correlator}, 
\begin{equation}
    \label{eq:gs-int-ES}
    \chi^{H_1}_E(\lambda) = 2 \widetilde{G}^{H_1}_0(\omega = 0) = 2 \int_0^\infty G^{H_1}_0(\tau ; \lambda) \dtau. 
\end{equation}
With these integral expressions now derived, we can compare the ES and FS,~\cref{eq:gs-int-ES,eq:gs-int-FS} more directly. Evidently, the important difference is that the FS involves a linear factor of $\tau$ in the integrand, which arises from the difference of the square in the denominators of~\cref{eq:spectral-ES,eq:spectral-FS}. Intuitively, one might expect that leads to the FS having a sharper peak at quantum critical points, and this holds rigorously by performing a standard scaling analysis~\cite{camposvenuti2007QuantumCriticalScaling, continentino2017quantum, gu2008FidelitySusceptibilityScaling} { and manifests in practical simulations~\cite{albuquerque2010QuantumCriticalScaling,tzeng2008fidelity}}. Considering $L^d$ particles on a $d-$dimensional lattice that undergoes a second-order QPT at critical point $\lambda_c,$ one can derive the scaling relations~\cite{albuquerque2010QuantumCriticalScaling},
\begin{align}
    \label{eq:es-scaling-relation}
    L^{-d} \chi^{H_1}_{E, 0}(\lambda) &\sim | \lambda - \lambda_c|^{-\alpha} \sim L^{(2/\nu) - (d+z)} \\
    \label{eq:fs-scaling-relation}
    L^{-d} \chi^{H_1}_{F,0}(\lambda) &\sim | \lambda - \lambda_c|^{-(2 - d \nu)} \sim L^{(2/\nu) - d}.
\end{align}
Here, the correlation length diverges as $\xi \sim | \lambda - \lambda_c|^{-\nu},$ $z$ is the dynamical critical exponent, and $\alpha$ is defined through the hyperscaling formula $2 - \alpha = \nu (d + z)$. Evidently, ES diverges whenever $\nu < 2 / (d+z)$ and the FS whenever $\nu < 2 / d$, and hence the FS is a stronger indicator of quantum criticality~\cite{albuquerque2010QuantumCriticalScaling}.  

These groundstate concepts have also been successfully extended to mixed states~\cite{zanardi2007InformationTheoreticDifferentialGeometrya,zanardi2007MixedstateFidelityQuantum,camposvenuti2007QuantumCriticalScaling}, particularly thermal states~\cite{zanardi2007BuresMetricThermal,albuquerque2010QuantumCriticalScaling,sirker2010FiniteTemperatureFidelitySusceptibility}. A thermal state of a given Hamiltonian $H$ is given by 
\begin{equation}
    \frac{1}{\mathcal{Z}} e^{-\beta H}, \ \ \ \mathcal{Z} = \Tr[e^{-\beta H}],
\end{equation}
for inverse temperature $\beta = 1 / T$. A thermal average is thus given by $\avg{O} \equiv \Tr[O e^{-\beta H}] / \mathcal{Z}$. We can now introduce a finite temperature connected correlator~\cite{albuquerque2010QuantumCriticalScaling}, 
\begin{equation}
    \label{eq:gibbs-correlator}
     G^{H_1}(\tau ; \lambda; \beta) = \theta(\tau) (\avg{H_1(\tau) H_1} - \avg{H_1}^2),
\end{equation}
with which we can derive the integral forms~\cite{albuquerque2010QuantumCriticalScaling},
\begin{align}
    \label{eq:gibbs-ES}
    \chi^{H_1}_{E}(\lambda; \beta) &= \int_0^\beta  G^{H_1}(\tau ; \lambda; \beta) \dtau \\
    \label{eq:gibbs-FS}
    \chi^{H_1}_F(\lambda; \beta) &= \int_0^{\beta/2} \tau G^{H_1}(\tau ; \lambda; \beta) \dtau.
\end{align}
These are particularly convenient for QMC~\cite{albuquerque2010QuantumCriticalScaling,schwandt2009QuantumMonteCarlo,wang2015FidelitySusceptibilityMade}. The finite temperature connected correlator is periodic about $\beta/2$, so we can equivalently write
\begin{equation}
    \chi^{H_1}_E(\lambda; \beta) = 2 \int_0^{\beta/2} G^{H_1}(\tau; \lambda; \beta) \dtau,
\end{equation}
for which it becomes apparent $ \chi^{H_1}_{E}(\lambda; \beta) \rightarrow  \chi^{H_1}_{E,0}(\lambda)$ and $\chi^{H_1}_{F}(\lambda; \beta) \rightarrow  \chi^{H_1}_{F,0}(\lambda)$ as $\beta \rightarrow \infty$. In this sense, the finite temperature susceptibilities are a strict generalization of the groundstate susceptibilities, so we will henceforth write only $\chi_E$ and $\chi_F$ when our discussion does not rely on specific parameters at play, i.e. choice of driving term, $H_1$. A longer discussion connecting~\cref{eq:gibbs-FS} to a more general notion of the Bures metric divergence at quantum critical points between mixed thermal states~\cite{zanardi2007BuresMetricThermal} is provided in Sec.~II.B of Ref.~\cite{albuquerque2010QuantumCriticalScaling}.

Given a Hamiltonian of the form~\cref{eq:qpt-ham}, we can estimate the ES and FS 
via \cref{eq:gibbs-ES,eq:gibbs-FS} provided we can estimate $\avg{H_1}, \int_0^\beta \avg{H_1(\tau) H_1}$ 
and $\int_0^{\beta/2} \tau \avg{H_1(\tau) H_1}$. 
In PMR-QMC, we can derive estimators for each of these quantities independent of the details of $H$, and hence, 
for arbitrarily complicated driving terms $H_1$. 
This generality is what allowed us to, using a single estimator code, estimate ES and FS 
for the variety of models in \cref{sec:numerics}. 
In the following, we justify this claim by presenting explicit estimators and discussing their properties and computational complexity. \\

\noindent \textbf{PMR-QMC background and notation}
 
A complete justification for our estimators requires substantial background in the permutation matrix representation (PMR) and PMR-QMC which we discuss in \cref{app:pmr-qmc-background} (see also Refs.~\cite{ezzell2025advanced,barash2024QuantumMonteCarlo}). The main derivations are also technical and lengthy and hence contained in \cref{app:best-estimator-derivations}. Here, we provide a summary mainly to define the necessary notation and basic ideas. 

Firstly, the PMR itself is a general matrix decomposition that allows us to write $H(\lambda) = D_0(\lambda) + \sum_{j=1}^M D_j(\lambda) P_j$ for $D_j(\lambda)$ matrices diagonal in the a fixed basis, $\{\ket{z}\}$ and $\{P_j\} \subset G$ for $G$ an Abelian permutation group with \emph{simple transitive group action} on $\{ \ket{z} \}$~\cite{ezzell2025advanced}.  Performing an off-diagonal series expansion of $\mathcal{Z} \equiv \Tr[e^{-\beta H(\lambda)}]$~\cite{hen2018OffdiagonalSeriesExpansion, albash2017OffdiagonalExpansionQuantum},  we can decompose the partition function into a sum of efficiently computable weights of the form~\cite{gupta2020CalculatingDividedDifferences,gupta2020PermutationMatrixRepresentation},
\begin{equation}
    \label{eq:pmr-z}
	\mathcal{Z} = \sum_{z, \Siq} w_{(z, \Siq)} \equiv \sum_{\mathcal{C}} w_{\mathcal{C}},
\end{equation}
for basis states $z$  and  products of permutations $\Siq = P_{\vec{i}_q} \cdot \ldots \cdot P_{\vec{i}_1}$ for $q \geq 0$ and each $P_{\vec{i}_k} \in \{P_j\}$, and $w_{\mathcal{C}}$.
Exploiting the fact that $G$ is Abelian
and only those terms contribute where $\Siq$ evaluates to the identity,
one can devise an automatic generation of a set of QMC update moves such that the resulting Markov chain is ergodic and satisfies detailed balance for essentially any physical Hamiltonian~\cite{barash2024QuantumMonteCarlo,akaturk2024quantum,babakhani2025quantum,babakhani2025fermionic}. This is the reason we are able to study a wide variety of models with a fixed code. 

A PMR-QMC estimator $O_{\mathcal{C}}$ is a function of $\mathcal{C}$ such that
\begin{equation}
    \label{eq:pmr-qmc-estimator}
    \avg{O} = \frac{\sum_{\mathcal{C}} w_{\mathcal{C}} O_{\mathcal{C}}}{\sum_{\mathcal{C}} w_{\mathcal{C}}},  
\end{equation}
which we denote as $O_{\mathcal{C}} \estimates \avg{O}$ to be read ``$O_{\mathcal{C}}$ estimates $\avg{O}$" for convenience.  Some of the present authors derived and explored the properties of a wide variety of estimators in PMR-QMC generally~\cite{ezzell2025advanced}. The basic idea is that since both $H$ and $O$ can always be decomposed in terms of the same group basis $G$,  then one can always derive an expression for $O_{\mathcal{C}}$.  In many cases of interest, such as if $O = H_0$ or $O = H_1$ for any bi-partition of $H$,  then this formal expression leads to a faithful estimator in practical simulation~\cite{ezzell2025advanced}.  This is why our estimators for the correlator,  ES,  and FS are independent of the choice of bi-partition.  

To discuss our estimators more in-depth,  we must introduce some standard PMR-QMC notation.  First,  we denote $\ket{z_k} \equiv S_{\vec{i}_k} \ket{z}$ as the basis state resulting from applying the first $k$ permutations from $\Siq$.  Next,  we denote the diagonal energies $E_{z_k} \equiv \avg{z_k | D_0 | z_k}.$ Writing $D_{\vec{i}_k}$ as the diagonal associated to $P_{\vec{i}_k}$,  we can define the off-diagonal hopping term,  $D_{(z,  \Siq)} \equiv \prod_k \avg{z_k | D_{\vec{i}_k} | z_k}$.  Finally,  we denote $e^{-\beta [E_{z_0}, \ldots, E_{z_q}]}$ as the divided difference of the exponential $e^{-\beta x}$ with respect to the multiset of diagonal energies given~\cite{deboor2005DividedDifferences, mccurdy1984accurate}. These thus defined,  we can write the standard PMR-QMC result
\begin{equation}
    \label{eq:pmr-qmc-weight}
	w_{(z, \Siq)} \equiv D_{(z, \Siq)} e^{-\beta [E_{z_0}, \ldots, E_{z_q}]},
\end{equation}
as the efficiently computable PMR-QMC weight. We remark that the divided difference of the exponential (DDE) is especially important in our estimator derivations, and a detailed explanation of the DDE is given in \cref{sec:div-diff}.  \\

\noindent \textbf{Pure diagonal estimators}

For simplicity, we first assume $H_1 \propto D_0$, i.e., as in the TFIM studied in \cref{sec:numerics}. Using the properties of the PMR-QMC, we can readily derive the estimators,
\begin{align}
    \label{eq:estimator-d0}
    E_z &\estimates \avg{D_0}
\end{align}
and
\begin{align}
    \label{eq:estimator-D0-corr}
    \begin{split}
        \frac{E_{z}}{e^{-\beta[E_{z_0}, \ldots, E_{z_q}]}} \sum_{j=0}^q E_{z_j} e^{-\tau [E_{z_j}, \ldots, E_{z_q}]}  \\ 
        \times \ e^{-(\beta - \tau)[E_{z_0}, \ldots, E_{z_j}]} &\estimates \avg{D_0(\tau) D_0}
    \end{split}
\end{align}
as we show in \cref{app:best-estimator-derivations}. These estimators involve only sums and products of (1) diagonal energies $E_{z_j}$ and (2) divided differences of these diagonal energies---the same quantities that appear in the PMR-QMC weight, \cref{eq:pmr-qmc-weight}. As such, they are both readily computable in a PMR-QMC simulation in $O(1)$ and $O(q^2)$ time, respectively (see Ref.~\cite{ezzell2025advanced} and \cref{app-sub:complexity-in-pmr-qmc} to understand the time complexity). Utilizing both of these estimators, we can estimate $G^{D_0}(\tau; \lambda; \beta)$ via Eq.~\eqref{eq:gibbs-correlator}.

Furthermore, the only non-trivial $\tau$ dependence in these estimators is contained in the divided difference expressions of \cref{eq:estimator-D0-corr}. Amazingly (see Ref.~\cite{zeng2025inequalities} and Lemma 1 in \cref{appsub:es-in-quadratic-time}), we can readily derive the exact integral relation
\begin{equation}
    \int_0^\beta e^{-\tau [x_{j+1}, \ldots, x_q]}  e^{-(\beta-\tau) [x_0, \ldots, x_j]} \dtau = - e^{-\beta [x_0, \ldots, x_q]}
\end{equation}
Combined with \cref{eq:estimator-D0-corr}, we have thus derived
\begin{multline}
    \label{eq:estimator-diag-es-int}
    \frac{- E_{z}}{e^{-\beta[E_{z_0}, \ldots, E_{z_q}]}} \sum_{j=0}^q E_{z_j} e^{-\beta[E_{z_0}, \ldots, E_{z_q}, E_{z_j}]} \\
    \estimates \int_0^\beta \avg{D_0(\tau) D_0} \dtau 
\end{multline}
which gives an estimator for $\chi_E^{D_0}$ (see Eq.~\eqref{eq:gibbs-ES}) when combined with \cref{eq:estimator-d0}.  As shown in \cref{app-sub:fs-in-quartic}, specifically Lemma 3, a similar but more complicated derivation yields,
   \begin{multline}
    \label{eq:estimator-diag-fs-int}
    \frac{E_z}{e^{-\beta[E_{z_0}, \ldots, E_{z_q}]}}\sum_{j=0}^q E_{z_j}  \sum_{r=0}^j e^{-\frac{\beta}{2}[E_{z_0}, \ldots, E_{z_r}]}  \\
    \times \sum_{m=j}^q e^{-\frac{\beta}{2}[E_{z_r}, \ldots, E_{z_q}, E_{z_j}, E_{z_m}]}
    \estimates \int_0^{\beta/2} \tau \avg{ D_0(\tau)D_0 \dtau }
\end{multline}
which gives an estimator for $\chi_F^{D_0}$ via Eq.~\eqref{eq:gibbs-FS}. These integrated quantities and hence the ES and FS can be evaluated in $O(q^2)$ and $O(q^4)$ time, respectively.

Several comments are in order. Firstly, our derivations for ES and FS involve first deriving an estimator for the correlator. We then extend this estimator to the ES and FS by an analytical and exact evaluation of the divided difference integrals, i.e., no numerical integration is required. Secondly, these specific estimators are for $H_1 \propto D_0$, but we can readily derive similar estimators for arbitrary $H_1$ (see ``General case estimators"). Even in the general case, the complexity to estimate the correlator, ES, and FS are $O(q^2), O(q^2)$ and $O(q^4)$, respectively. Thirdly, the complexity for the ES and FS can be improved in certain special cases as we discuss next. \\

\noindent \textbf{Improved ES and FS estimators}

In the previous section, we derived $O(q^2)$ and $O(q^4)$ estimators for the ES and FS for $H_1 \propto D_0$. These estimators can be improved to $O(1)$ and $O(q^3)$ time in this case (as well as if $H_1 \propto H - D_0)$, which we state as theorems. 

\begin{theorem}[An $O(1)$ ES estimator]
    \label{thrm:constant-ES-estimator}
    The ES for $H_1 \propto D_0$ or $H_1 \propto H - D_0$ can be estimated in $O(1)$ time via  
    \begin{multline}
        \label{eq:ES-estimator-new}
        \beta E_{z_0}^2 + q E_{z_0} + \mathbf{1}_{q>0} \beta E_{z_0} \frac{ e^{-\beta [E_{z_0}, \ldots, E_{z_{q-1}}] } }{ e^{-\beta [E_{z_0}, \ldots, E_{z_q}]} } \\ 
        \estimates \int_0^\beta \avg{ D_0(\tau) D_0 } \dtau,
    \end{multline}
   for $\mathbf{1}_{q>0}$ zero if $q \leq 0$ and one if $q > 0$. 
\end{theorem}

\begin{theorem}[An $O(q^3)$ estimator for FS]
    \label{thrm:cubic-FS-estimator}
    The FS for $H_1 \propto D_0$ or $H_1 \propto H - D_0$ can be computed in $O(q^3)$ time via 
    \begin{multline}
        \label{eq:FS-estimator-simp1}
        \mathbf{1}_{q=0} \frac{(\beta E_{z_0})^2}{8} + \mathbf{1}_{q>0} \frac{E_{z_0}}{e^{-\beta [E_{z_0}, \ldots, E_{z_q}]}} \sum_{r = 0}^q e^{-\frac{\beta}{2} [E_{z_0}, \ldots, E_{z_r}]}  \\
        \times \Bigg( \frac{\beta^2}{8} E_{z_r} e^{-\frac{\beta}{2} [E_{z_r}, \ldots, E_{z_q}]} +  \frac{\beta^2}{8} e^{-\frac{\beta}{2} [E_{z_{r+1}}, \ldots, E_{z_q}]} - \\
         \sum_{i=r+1}^q (i - r)  e^{-\frac{\beta}{2} [E_{z_r}, \ldots, E_{z_q}, E_{z_i}]}  \Bigg) \estimates \int_0^{\beta/2} \tau \avg{ D_0(\tau) D_0 } \dtau,
    \end{multline}
   for $\mathbf{1}_{q>0}$ zero if $q \leq 0$ and one if $q > 0$. 
\end{theorem}

Our proof, given in complete detail in \cref{app:best-estimator-derivations}, amounts to stating and proving several novel divided difference relations that can be used to simplify \cref{eq:estimator-diag-es-int} and \cref{eq:estimator-diag-fs-int}. While this naively only simplifies the $H_1 \propto D_0$ case, the $H_1 \propto H - D_0$ case follows by linearity and the derivation of a few $O(1)$ estimators for average energy $\avg{H}$ and similar quantities. 

Stated abstractly, the constraint $H_1 \propto D_0$ or $H_1 \propto H - D_0$ may seem restrictive, but recall that this applies to the TFIM and XXZ models studied in \cref{sec:numerics}. More generally, many well studied models in condensed matter physics have such simple driving terms that are either purely diagonal or purely off-diagonal. As such, we have derived a constant time ES estimator for a wide variety of models. \\

\noindent \textbf{General case estimators}

Generally, $H_1$ may contain both diagonal and off-diagonal terms, i.e., \cref{eq:prl-model} or the random ensemble studied in \cref{sec:numerics}. This is not a fundamental issue for PMR-QMC, as whenever $H$ is in PMR form with $R$ term, $H = \sum_{j} D_j P_j,$ we can always write $H_1 = \sum_{j} \Lambda_j D_j P_j$ for $\Lambda_j$ diagonal matrices by the group structure of the PMR (see \cref{app:pmr-form-of-h1}). Observables of this type are said to be in canonical PMR form, and can be estimated correctly~\cite{ezzell2025advanced}. Namely, we can readily estimate $\avg{H_1}$ and  $\avg{H_1(\tau) H_1}$ using expressions that generalize \cref{eq:estimator-d0,eq:estimator-D0-corr}.

For example, we can readily derive the estimator~\cite{ezzell2025advanced},
\begin{equation}
    \label{eq:laml-dl-pl-estimator}
    \delta_{P_l}^{(q)} \frac{\Lambda_l(z) e^{-\beta [E_{z_0}, \ldots, E_{z_{q-1}}]} }{e^{-\beta [E_{z_0}, \ldots, E_{z_q}]}} \estimates \avg{\Lambda_l D_l P_l},
\end{equation}
from which we can estimate $\avg{H_1}$ by linearity. Here, $\delta_{P_l}^{(q)}$ equals $1$ if the $q^{\text{th}}$ or final permutation in $\Siq$ is $P_l$ and is $0$ otherwise, and we have used the shorthand $\Lambda_l(z) \equiv \avg{z | \Lambda_l | z}.$ As for $\avg{H_1(\tau) H_1}$, we can also derive an estimator for $\avg{(\Lambda_k D_k P_k)(\tau) \Lambda_l D_l P_l}$ that---up to small technicalities---contains the same $\tau$ dependence as \cref{eq:estimator-D0-corr} used to estimate $\avg{D_0(\tau) D_0}.$ Hence, even for arbitrary $H_1$, we can extend our estimator for $\avg{H_1(\tau) H_1}$ to estimators for $\chi_E^{H_1}$ and $\chi_F^{H_1}$ by our novel divided difference integral relations alone. 

Specific expressions for these estimators and more discussion are provided in \cref{app:general-driving-term} for interested readers. The functional difference between these general estimators and those for $H_1 \propto D_0$ or $H_1 \propto H - D_0$ is that we estimate things term-by-term and exploit linearity. As such, the complexity estimation also depends on the number of non-trivial terms contained in $H_1$. Assuming there are $R$ non-trivial terms, the complexity of estimating the ES and FS in the most general case is then $O(R^2 q^2)$ and $O(R^2 q^4)$, respectively. In practice, the average complexity is much better, as we do not need to compute the full estimators whenever various $O(1)$ delta function conditions are not satisfied (see \cref{eq:laml-dl-pl-estimator}). Consequently, we were able to, for example, study the random ensemble with hundreds of terms successfully in \cref{sec:numerics}.  \\

\section{Comparison to prior estimators}

We compare our work to similar efforts in Refs.~\cite{schwandt2009QuantumMonteCarlo,albuquerque2010QuantumCriticalScaling,wang2015FidelitySusceptibilityMade} which heavily inspired us. The principal advantage of our formulation is its direct applicability to arbitrary Hamiltonians and driving terms without modification, enabling, for example, the study of random ensembles in \cref{fig:prl-rot-obs}. For a more direct comparison, we henceforth focus on conventional models for which established (finite temperature) QMC algorithms exist and where the driving term is purely off-diagonal, $H_1 \propto H - D_0$ (i.e.,as in as in XXZ). In this case, we have derived $O(1)$ ES and $O(q^3)$ FS estimators, respectively. 

In Ref.~\cite{wang2015FidelitySusceptibilityMade}, a general FS estimator applicable to multiple QMC frameworks is introduced, which improves upon the prior stochastic series expansion (SSE) estimator derived in Refs.~\cite{schwandt2009QuantumMonteCarlo,albuquerque2010QuantumCriticalScaling}. Given the close structural relationship between SSE and PMR-QMC~\cite{albash2017OffdiagonalExpansionQuantum,hen2018OffdiagonalSeriesExpansion,gupta2020PermutationMatrixRepresentation}, we compare our estimator to the SSE estimator derived in Ref.~\cite{wang2015FidelitySusceptibilityMade}. For an SSE operator string truncation length $M$, the estimator (see Fig.~3 of \cite{wang2015FidelitySusceptibilityMade}) requires $O(M)$ computational effort in the worst case.
Since this is as fast as the $O(M)$ estimator for ES derived in Ref.~\cite{albuquerque2010QuantumCriticalScaling}, one need not consider measuring ES in SSE. Although $M$ is not directly comparable to $q$ in PMR-QMC—each PMR-QMC weight represents an infinite resummation of SSE terms through the divided-difference treatment of diagonal contributions~\cite{albash2017OffdiagonalExpansionQuantum,gupta2020PermutationMatrixRepresentation,hen2018OffdiagonalSeriesExpansion}—one may expect the $O(M)$ estimator to be asymptotically faster than our $O(q^3)$ FS estimator.

The practical advantage of an $O(M)$ SSE estimator over our $O(q^3)$ PMR-QMC estimator is not necessarily decisive, however. First, the total simulation time in QMC is not always dominated by estimator evaluation. As illustrated in \cref{fig:maintext-timing-2}, simulation time is often governed by the Monte Carlo updates themselves rather than by measurement routines, and similar conditions were used to obtain the ES and FS data in \cref{fig:prx-replicate-fig5}. In such cases, differences in estimator complexity are largely inconsequential. Second, PMR-QMC has been observed to equilibrate substantially faster than SSE for disordered or glassy systems~\cite{albash2017OffdiagonalExpansionQuantum,gupta2020PermutationMatrixRepresentation}, again diminishing the relevance of raw estimator complexity. 

We therefore expect a meaningful advantage of the $O(M)$ SSE estimator only in regimes where the total computational cost of PMR-QMC is indeed dominated by our $O(q^3)$ estimator (we explore this regime of PMR-QMC empirically in \cref{app:add-empirical-resources}). Yet even in these cases, one can alternatively compute the ES, for which we have derived an $O(1)$ estimator. Unless one is studying a system where the FS exhibits a pronounced peak in the absence of a corresponding ES feature~\cite{tzeng2008fidelity}, it is unclear which approach offers superior overall efficiency. Accordingly, we conclude that the primary advantage of the formulation in Ref.~\cite{wang2015FidelitySusceptibilityMade} lies in its generality across diverse QMC methodologies.

\vspace{-0.23cm}
\section*{Code and Data Availability}
Simulation code, data, and analysis scripts are open source~\cite{ezzell2025code} and available either through Zenodo,
\href{https://zenodo.org/records/17786518}{zenodo.org/records/17786518}, or GitHub, \href{https://github.com/naezzell/advmeaPMRQMC/tree/v3.0.0}{github.com/naezzell/advmeaPMRQMC/tree/v3.0.0}. 

\section*{Acknowledgements} 
The authors acknowledge the Center for Advanced Research Computing (CARC) at the
University of Southern California for providing the computing resources
used in this work. This material is based upon work supported by the Defense Advanced Research Projects Agency (DARPA) under Contract No. HR001122C0063. All material, except scientific articles or papers published in scientific journals, must, in addition to any notices or disclaimers by the Contractor, also contain the following disclaimer: Any opinions, findings and conclusions or recommendations expressed in this material are those of the author(s) and do not necessarily reflect the views of the Defense Advanced Research Projects Agency (DARPA). N.E. was partially supported by the U.S. Department of Energy
(DOE) Computational Science Graduate Fellowship under
Award No. DE-SC0020347 and the ARO MURI grant W911NF-22-S-000 during parts of this work. 

\vspace{-0.25cm}
\section*{Author Contributions}
I.H. and N.E. conceived the research problem.
N.E., L.B., and I.H. jointly derived the analytical results for the estimators.
N.E. and L.B. developed the software, performed the numerical experiments, and analyzed the data.
All authors wrote and revised the manuscript and approved the final version.

\section*{Competing Interests}
The authors declare no competing interests.

\bibliography{refs}

\clearpage
\onecolumngrid
\appendix
\setcounter{secnumdepth}{2}

\section{Review of technical PMR-QMC background}
\label{app:pmr-qmc-background}

\subsection{Permutation matrix representation}
\label{app:pmr}

The permutation matrix representation (PMR)~\cite{gupta2020elucidating} is a general matrix decomposition of square matrices~\cite{ezzell2025advanced} with several applications in condensed matter physics and quantum information~\cite{gupta2020PermutationMatrixRepresentation, gupta2020elucidating, hen2021determining, akaturk2024quantum, barash2024QuantumMonteCarlo, chen2021quantum, kalev2021integral, kalev2021quantum, kalev2024feynman,babakhani2025quantum,babakhani2025fermionic}. Functionally, we say a square matrix $H$ is in PMR form provided it is written $H = D_0 + \sum_j D_j P_j$ for each $D_j$ a diagonal matrix and $P_j$ permutations with no fixed points~\cite{gupta2020elucidating}. More rigorously~\cite{ezzell2025advanced}, this is guaranteed provided $\{P_j\} \subset G$ for $G$ a permutation group with \emph{simply transitive group action} on computational basis states, $\{\ket{z}\}$. The most important consequence of this form is that products of permutations $P_i P_j$ have no fixed points unless $P_i P_j = \mathds{1}$. A through definition of the PMR with several examples is contained in the work by some of the present authors~\cite{ezzell2025advanced}. 

\subsection{Off-diagonal series expansion}
\label{app:pmr-ode}
The off-diagonal series expansion (OSE), first introduced in Refs.~\cite{hen2018OffdiagonalSeriesExpansion,albash2017OffdiagonalExpansionQuantum}, is the starting point of the PMR-QMC~\cite{gupta2020elucidating}. Alongside basic group theory, it is also the tool by which we can readily derive general estimators for arbitrary observables in PMR-QMC~\cite{ezzell2025advanced}. Let $f$ by any analytic function and $H = D_0 + \sum_j D_j P_j$ be a Hamiltonian in PMR form. The OSE, then, is a convenient expansion of $\Tr[f(H)]$ into a ``classical term" and an infinite series of ``quantum, non-commuting" corrections, which will become clear shortly. 

A derivation of the OSE of $\Tr[f(H)]$ begins by a direct Taylor series expansion of $f(H)$ about 0. Doing so yields, 
\begin{align}
    \Tr[f(H)] &= \sum_z \sum_{n=0}^\infty \frac{f^{(n)}(0)}{n!} \avg{z | ( D_0 + \sum_j D_j P_j )^n | z} \\
    &= \sum_z \sum_{n=0}^\infty \sum_{\{C_{\vec{i}_n}\}} \frac{f^{(n)}(0)}{n!} \avg{ z | C_{\vec{i}_n} | z},
\end{align}
where in the second line we sum over all operator sequences consisting of $n$ products of $D_0$ and $D_j P_j$ terms, which we denote $\{C_{\vec{i}_n}\}$. The multi-index $\vec{i}_n \equiv (i_1, \ldots, i_n)$ denotes the ordered sequence, i.e. $\vec{i}_3 = (3, 0, 1)$ indicates the sequence $C_{\vec{i}_3} = (D_1 P_1) D_0 (D_3 P_3)$, read from right to left.  More generally, each $i_k \in \{0, \ldots, M\}$ denotes a single term from the PMR form $H = \sum_{j=0}^M D_j P_j$.

For convenience, we can separate the contributions from diagonal operators $D_j$ from off-diagonal permutations, $P_j$ which yields~\cite{albash2017OffdiagonalExpansionQuantum, hen2018OffdiagonalSeriesExpansion} the following complicated expression,
\begin{equation}
    \label{eq:initial-tr-expansion}
    \Tr[f(H)] = \sum_z \sum_{q=0}^{\infty} \sum_{\Siq} D_{(z, \Siq)} \avg{ z | \Siq | z} \left( \sum_{n=q}^{\infty} \frac{f^{(n)}(0)}{n!} \sum_{\sum_i k_i = n - q} E^{k_0}_{z_0} \cdot \ldots \cdot E^{k_q}_{z_q} \right),
\end{equation}
which is justified in prior works~\cite{albash2017OffdiagonalExpansionQuantum, hen2018OffdiagonalSeriesExpansion}. We now briefly summarize the notation, which will be used throughout. Firstly, $\Siq = P_{\vec{i}_q} \cdot \ldots \cdot P_{\vec{i}_1}$ denotes a product over $q$ permutations, each taken from $\{P_j\}_{j=1}^M$, i.e., the multiset indices are now $\vec{i}_j \in \{1, \ldots, M\}$. Next, we denote $\ket{z_0} \equiv \ket{z}$ and $\ket{z_k} \equiv P_{\vec{i}_k} \cdot \ldots \cdot P_{\vec{i}_1} \ket{z}$. This allows us to define the ``diagonal energies" as $E_{z_k} \equiv \avg{z_k | H | z_k} = \avg{z_k | D_0 | z_k}$ (recall $\ket{z_k}$ is a basis for $D_0$ not of $H$) and the off-diagonal ``hopping strengths" $D_{(z, \Siq)} \equiv \prod_{k=1}^q \avg{z_k | D_{\vec{i}_k} | z_k}$. The sum over $\sum_{\sum k_i}$ is again the set of weak partitions of $n - q$ into $q+1$ integers.

This complicated expression can be greatly simplified in two ways. Firstly, by the no-fixed points property of the PMR~\cite{ezzell2025advanced}, $\avg{z | \Siq | z} = 1$ if and only if $\Siq = \mathds{1}$ and is zero otherwise. Secondly, as first noticed in earlier works~\cite{albash2017OffdiagonalExpansionQuantum, hen2018OffdiagonalSeriesExpansion}, the diagonal term in parenthesis can be identified with the so-called divided difference of $f$~\cite{deboor2005DividedDifferences,mccurdy1984accurate}. Together, we can write 
\begin{equation}
    \label{eq:off-diagonal-series-expansion}
    \Tr[f(H)] = \sum_z \sum_{q=0}^{\infty} \sum_{\Siq = \mathds{1}} D_{(z, \Siq)}  f[E_{z_0}, \ldots, E_{z_q}],
\end{equation}
for $f[E_{z_0}, \ldots, E_{z_q}]$ is the divided difference of $f$ with respect to the multiset of diagonal energies $E_{z_0}, \ldots, E_{z_q}$. The divided difference plays a pivotal role in our derivation, so we review it in detail in the next section. For now, we simply remark that this form is precisely what we refer to as the off-diagonal series expansion of $\Tr[f(H)]$. As the name suggests, it is a perturbative series in $q$, the size of the off-diagonal ``quantum dimension." Put differently, when $q = 0$, this is the expansion of $\Tr[f(D_0)]$, and the remaining terms correct for the off-diagonal, non-commuting contribution, as claimed in the beginning. 

\subsection{Divided differences}
\label{sec:div-diff}

We briefly review the technical details of the divided difference,  inspired by the treatment in Refs.~\cite{albash2017OffdiagonalExpansionQuantum, hen2018OffdiagonalSeriesExpansion, ezzell2025advanced}.  The divided difference of any holomorphic function $f(x)$ can be defined over the multiset $[x_0, \ldots, x_q]$ using a contour integral~\cite{mccurdy1984accurate, deboor2005DividedDifferences},
\begin{equation}
    \label{eq:contour-int-dd}
    f[x_0, \ldots, x_q] \equiv \frac{1}{2\pi i} \oint\limits_{\Gamma} \frac{f(x)}{\prod_{i=0}^q(x - x_i)} \mathrm{d}x,
\end{equation}
for $\Gamma$ a positively oriented contour enclosing all the $x_i$'s. Several elementary properties utilized in PMR-QMC~\cite{albash2017OffdiagonalExpansionQuantum, gupta2020PermutationMatrixRepresentation} follow directly from this integral representation, or by invoking Cauchy's residue theorem. For example, $f[x_0, \ldots, x_q]$ is invariant to permutations of arguments, the definition reduces to Taylor expansion weights when arguments are repeated,
\begin{equation}
    f[x_0, \ldots, x_q] = f^{(q)}(x) / q!, \ \ \ x_0 = x_1 \ldots = x_q = x,
\end{equation}
and whenever each $x_i$ is distinct, we find
\begin{equation}
    f[x_0, \ldots, x_q] = \sum_{i=0}^q \frac{f(x_i)}{ \prod_{k \neq i} (x_i - x_k)},
\end{equation}
the starting definition in Refs.~\cite{albash2017OffdiagonalExpansionQuantum, hen2018OffdiagonalSeriesExpansion, gupta2020PermutationMatrixRepresentation}. Each of these divided difference definitions can be shown to satisfy the Leibniz rule~\cite{deboor2005DividedDifferences},
\begin{equation}
    \label{eq:leibniz-rule}
    (f \cdot g)[x_0, \ldots, x_q] = \sum_{j=0}^q f[x_0, \ldots, x_j] g[x_j, \ldots, x_q] = \sum_{j=0}^q g[x_0, \ldots, x_j] f[x_j, \ldots, x_q]
\end{equation}
which is particularly important in our derivations. 

Yet another useful way to view the divided difference for our work is to derive its power series expansion,
\begin{align}
    \label{eq:dd-as-series-expansion}
    f[x_0, \ldots, x_q] = \sum_{m=0}^\infty \frac{f^{(q+m)}(0)}{(q+m)!} \sum_{\sum k_j = m} \prod_{j=0}^q x_j^{k_j},
\end{align}
where the notation $\sum_{\sum k_j = m}$ is a shorthand introduced in Refs.~\cite{hen2018OffdiagonalSeriesExpansion, albash2017OffdiagonalExpansionQuantum} which represents a sum over all \emph{weak integer partitions} of $m$ into $q+1$ parts. More explicitly, it is an enumeration over all vectors $\vec{k}$ in the set $\{\vec{k} = (k_0, \ldots, k_q) : k_i \in \mathds{N}_0, \sum_{j=0}^q k_j = m\}$, where $\mathds{N}_0$ is the natural numbers including zero. To derive the series expansion, we first Taylor expand $f(x)$ inside the contour integral, 
\begin{align}
    f[x_0, \ldots, x_q] &= \sum_{n=0}^\infty \frac{f^{(n)}(0)}{n!} \left( \frac{1}{2 \pi i } \oint_{\Gamma} \frac{x^n}{\prod_{i=0}^q (x - x_i)} \right) 
    = \sum_{n=0}^\infty \frac{f^{(n)}(0)}{n!}  [x_0, \ldots, x_q]^n ,
\end{align}
where we have introduced the shorthand $[x_0, \ldots, x_q]^k \equiv p_k[x_0, \ldots, x_q]$ for $p_k(x) = x^k$ as introduced in Refs.~\cite{hen2018OffdiagonalSeriesExpansion, albash2017OffdiagonalExpansionQuantum}. The divided difference of a polynomial has a closed form expression most easily written with the change of variables $n \rightarrow q + m$,
\begin{equation}
    \label{eq:poly-div-diff}
    [x_0, \ldots, x_q]^{q+m} = 
    \begin{cases}
    0 & m < 0 \\
    1 & m = 0 \\
    \sum_{\sum_{k_j} = m} \prod_{j=0}^q x_j^{k_j} & m > 0,
    \end{cases}
\end{equation}
as noted in the same notation in Refs.~\cite{hen2018OffdiagonalSeriesExpansion, albash2017OffdiagonalExpansionQuantum} and derived with a different notation in Ref.~\cite{deboor2005DividedDifferences}. Together, these two observations yield~\cref{eq:dd-as-series-expansion}.

As with previous PMR works~\cite{albash2017OffdiagonalExpansionQuantum, gupta2020PermutationMatrixRepresentation}, we are also especially interested in the divided difference of the exponential (DDE) where we utilize the shorthand notation $e^{t [x_0, \ldots, x_q]} \equiv f[x_0, \ldots, x_q]$ for $f(x) = e^{t x}$. Replacing $f^{(q+m)}(0)$ with $t^{q+m}$ gives the power series expansion of the DDE, as derived in the appendices of Refs.~\cite{albash2017OffdiagonalExpansionQuantum, hen2018OffdiagonalSeriesExpansion}. Replacing the variable $x \rightarrow \alpha x$ in~\cref{eq:contour-int-dd}, we find the rescaling relation,
\begin{equation}
    \label{eq:rescaling-relation}
    \alpha^q e^{t [\alpha x_0, \ldots, \alpha x_q]} = e^{\alpha t [x_0, \ldots, x_q]},
\end{equation}
which is useful in numerical schemes~\cite{gupta2020CalculatingDividedDifferences, barash2022calculating} and in computing the Laplace transform. In particular, combining~\cref{eq:rescaling-relation} with Eq. (11) of Ref.~\cite{kunz1965inverse} with $P_m(x) = 1$, we find
\begin{equation}
    \label{eq:laplace-of-dd}
    \lap{e^{\alpha t [x_0, \ldots, x_q]}} = \frac{\alpha^q}{ \prod_{j=0}^q (s - \alpha x_j)},
\end{equation}
where $\mathcal{L}$ denotes the Laplace transform from $t \rightarrow s$. One can alternatively derive~\cref{eq:laplace-of-dd} by directly performing the integration to the contour integral definition in~\cref{eq:contour-int-dd}, e.g. by Taylor expanding $e^{\alpha t x}$ and re-summing term-by-term, legitimate by the uniform convergence of the exponential.

\subsection{Computational complexity of estimators in PMR-QMC}
\label{app-sub:complexity-in-pmr-qmc}

Throughout this work, we refer to complexity of evaluating PMR-QMC estimators. For interested readers, more details and examples of tracking complexity in PMR-QMC estimation are contained in Sec.~VIII of Ref.~\cite{ezzell2025advanced}. Here, we give a brief summary of how time is accounted. 

When evaluating an estimator in PMR-QMC, we have $O(1)$ access to each $E_z$ and to each 
$\{\exp{-\beta [E_{z_0}]}$, $\exp{-\beta [E_{z_0}, E_{z_1}]}$, $\ldots$, $\exp{-\beta [E_{z_0}, \ldots, E_{z_{q-1}}]}$, $\exp{-\beta [E_z, \ldots, E_{z_q}]}\}$~\cite{gupta2020PermutationMatrixRepresentation, barash2024QuantumMonteCarlo}, but evaluating the divided difference with the addition or removal of an input has an $O(q)$ cost~\cite{gupta2020CalculatingDividedDifferences}.
For example, $E_z$ is an O$(1)$ estimator of $\avg{D_0}$ by Eq.~\eqref{eq:estimator-d0} and Eq.~\eqref{eq:estimator-D0-corr} is an $O(q^2)$ estimator of $\avg{D_0(\tau) D_0}$. Generally, we  expect $\avg{q} \propto \beta$ for many systems~\cite{albash2017OffdiagonalExpansionQuantum}, but we also expect the total simulation time to be determined by the complexity of QMC updates rather than that of the estimators for a wide range of temperatures. 
{
\section{Code usage and empirical resource scaling}
\label{app:code-and-empirical-resource-scaling}

The open-source \verb^C++^ implementation accompanying this work~\cite{ezzell2025code} extends previously developed PMR-QMC frameworks capable of simulating arbitrary spin-$1/2$ Hamiltonians~\cite{barash2024QuantumMonteCarlo, barash2024PmrQmcCode}. The software has been designed to facilitate reproducibility and ease of use. In this section, we summarize the basic usage workflow and provide representative empirical benchmarks of computational resource requirements.

\subsection{Code usage}
\label{app:code-usage}

Comprehensive usage instructions and example input files are provided in the accompanying \verb^README^ file~\cite{ezzell2025code}. Here, we provide a specific simple example to illustrate the ease with which one may use our code to estimate ES and FS for a user-defined Hamiltonian.

First, we define the Hamiltonian as a human-readable Pauli string in a file called \verb^H.txt^. For concreteness,
\begin{verbatim}
-1.00 1 Z 2 Z 3 Z
-0.68 1 X
-0.73 1 Z 2 X
0.82 1 Y 3 Y
\end{verbatim}
defines \verb^H.txt^ as a sum of random Pauli strings. 
To specify simulation parameters and what we intended to measure, we edit a file called \verb^parameters.hpp^. For example,
\begin{verbatim}
#define beta 2.0 
#define steps 100000 
\end{verbatim}
sets $\beta = 2.0$ and the number of MC steps to $10^5$. We can also specify which measurements to make, i.e., 
\begin{verbatim}
#define MEASURE_HDIAG
#define MEASURE_HDIAG_CORR
#define MEASURE_HDIAG_EINT
#define MEASURE_HDIAG_FINT
\end{verbatim}
informs our code to measure $\avg{D_0}$, $\avg{D_0(\tau) D_0}$ (for $\tau$ specified in the parameter file), $\int_0^\beta \avg{D_0(\tau) D_0} \dtau$ and $\int_0^{\beta/2} \tau \avg{D_0(\tau) D_0} \dtau$. Note that if our code detects that one intends to estimate both \verb^HDIAG^ and \verb^HDIAG_EINT^, it automatically compute the ES via Eq.~\eqref{eq:gibbs-ES} via jackknife binning analysis~\cite{berg2004introduction}. Similarly, the FS is estimated via Eq.~\cref{eq:gibbs-FS} provided  \verb^HDIAG^ and \verb^HDIAG_FINT^ are measured. 

We then compile and execute our \verb^C++^ code, specifying whether we intend to perform parallelization with MPI or not, and if so, how many threads (approximately the number of CPU cores) to use. Compiling and executing the above example on our laptop using six threads, $\nthreads = 6$, yields a simulation summary text file whose final lines include,
\begin{verbatim}

Number of MPI processes: 6

Collecting statistics and finalizing the calculation

Total number of MC updates = 600000
Total mean(q) = 3.365
Total max(q) = 14
Total mean(sgn(W)) = 0.674666667
Total std.dev.(sgn(W)) = 0.00939374389
Total of observable #1: H_{diag}
Total mean(O) = -0.800449775
Total std.dev.(O) = 0.0133078895
Total of observable #2: measure_Hdiag_corr
Total mean(O) = 0.727002872
Total std.dev.(O) = 0.0109234895
Total of observable #3: measure_Hdiag_Eint
Total mean(O) = 1.58395933
Total std.dev.(O) = 0.015689912
Total of observable #4: measure_Hdiag_Fint
Total mean(O) = 0.377212308
Total std.dev.(O) = 0.0048761732
Total of derived observable: diagonal energy susceptibility
Total mean(O) = 0.302519643
Total std.dev.(O) = 0.0330613162
Total of derived observable: diagonal fidelity susceptibility
Total mean(O) = 0.0568523869
Total std.dev.(O) = 0.00782769997
Total elapsed cpu time = 0.894016 seconds

Wall-clock time = 0.174136 seconds
\end{verbatim}

which gives the estimated value of each of the above observables, estimated over six MPI threads. 

\subsection{Simulation parameter inputs and output statistics}
Out software supports changing the following basic simulation parameters,
\begin{enumerate}
    \item $\beta$ -- the inverse temperature
    \item $\Tsteps$ -- number of QMC moves used to thermalize/equilibrate (i.e., before any measurements are taken)
    \item $\steps$ -- number of QMC moves during which measurements are taken
    \item $\stepsPer$ -- number of QMC moves between measurements (i.e., controls the sparsity of measurements)
    \item $\nbins$ -- number of bins for the error estimation via binning analysis (see Appendix B of Ref.~\cite{barash2024QuantumMonteCarlo})
    \item $\nthreads$ -- number of MPI threads (approximately the number of CPU cores) used in lazy parallel accumulation, 
\end{enumerate}
all of which affect the expected total wall-clock time required for desired observable estimates to converge to thermal expectation values. Of particular note, $\stepsPer$ can control whether the QMC simulation time is dominated by QMC update costs or QMC measurements costs, as we will return to in the next section. In the above example, we chose $\beta = 2$, $\Tsteps = 0$, $\steps = 10^5$, $\stepsPer = 100$, $\nbins = 100$, and $\nthreads = 6$. 

We also briefly mention some of the helpful outputs of our algorithm that are not just observable estimates. For all simulations, the statistics,
\begin{itemize}
    \item $\qavg$ -- the empirical average value of $q$ throughout  simulation
    \item $\qmax$ -- the maximum instantaneous value of $q$ throughout simulation
    \item $\avgsign$ -- the average sign of the PMR-QMC weight throughout simulation
    \item $\sigsign$ -- the standard deviation of $\text{sgn}$ across the simulation 
    \item $\sigma_O$ -- one standard deviation error estimate for measured observable $O$, obtained by binning analysis~\cite{barash2024QuantumMonteCarlo}
\end{itemize}
are returned automatically. For the above example, we find $\qavg = 3.365$, $\qmax = 14$, $\avgsign = 0.675$, and $\sigsign = 0.009$. Provided at least 5 MPI threads are used, our simulation also automatically performs autocorrelation/thermalization testing. The testing principle is quite simple: each parallel simulation runs with a different random seed and is hence uncorrelated from one another after a sufficient number of  burn-in steps, $\Tsteps$. One can thus compare the statistical error obtained from any single run to that of the error estimated across $\nthreads$ runs to check for autocorrelation. 

Namely, if the two error estimates are comparable in value, this suggests that measurement blocks are effectively uncorrelated and thus that autocorrelation time is shorter than the block length. If the two estimates differ, albeit not substantially, this typically suggests that the autocorrelation time exceeds the block length, yet remains much shorter than the timescale of the entire simulation. Yet when the error from a single run significantly underestimates the error across runs (e.g., by a factor of 2 or more), this typically indicates that the autocorrelation time is too large, and longer simulations are required to ensure reliable estimates. In the above example, the thermalization test for $\avg{D_0}$ reads
\begin{verbatim}
Observable #1: H_{diag}, mean of std.dev.(O) = 0.0353, std.dev. of mean(O) = 0.0409: test passed
\end{verbatim}
which shows that the test passes. In fact, thermalization tests for all observables passed.

\subsection{Additional empirical resource estimates}
\label{app:add-empirical-resources}

We provide additional empirical resource estimates and information, supplementing our brief summary in \cref{sec:numerics}. In \cref{tab:app_simulation_parameters}, we summarize in greater detail (compared to \cref{tab:simulation_parameters}) the simulation  input parameters, resources, and output statistics relevant to the simulation used to generate data for the main figures in our paper. As stated in the main text, 200MB of memory was more than enough for our simulations, so we focus on simulation time as the important resource estimate. 

\begin{table*}[htp!]
\centering
\begin{tabular}{lcccccc}
\toprule
Figure & $\Tsteps$ & $\steps$ & $\stepsPer$ & $\nthreads$  & $\walltime$ & $\qavg$ \\
\midrule
1 & $2\times10^7$ & $10^6$ & $10$ & 5 & 10\,s -- 1\,min & $0.10{-}2.78$  \\
2 & $10^6{-}10^7$ & $10^7{-}5\times10^8$ & $10^3$ & 100 & 15\,s -- 14.4\,h & $0.33{-}119.87$ \\
3 & $10^6{-}10^7$ & $3\times10^7{-}2\times10^9$ & $5\times10^3{-}4\times10^4$ & 200--400 & 3.6\,h -- 8.8\,days & $23.58{-}328.83$ \\
4 & $10^5$ & $10^6$ & $10$ & 5 & 4\,s -- 8.2\,min & $0.34{-}3.31$ \\
\bottomrule
\end{tabular}
\caption{ Simulation parameters, wall-clock times, and average value of $q$ corresponding to Figs.~1, 2, 3, and 4. Ranges reflect the minimum and maximum value for any listed parameter across entire dataset. All simulations required less than 200\,MB of memory per core. Except for the mild sign problem encountered in the model of Fig.~4 (with worse average sign for any parameter studied given by $\avgsign = 0.85(4)$), all other models were sign-problem free under PMR-QMC.}
\label{tab:app_simulation_parameters}
\end{table*}

This table should be interpreted as a coarse estimate of simulation time (and not absolute or optimal times) needed to study the critical properties of these models for several reasons. First, we did not attempt to optimize these parameters for speed, but instead simply choose values large enough for correctness. In addition, we measured all observables of interest at the same time, so times are reflective of the slowest thermalization observable. Third, we did not keep hardware constant, as we ran simulations on a high performance computing (HPC) cluster with a queuing system. 

Nevertheless, certain informative relative scalings can be observed from our simulations. Visually, this is perhaps most appealing seen in \cref{fig:app-prl-timing-study}. Here, is is seen that criticality indicators ES and FS follow a similar trend to total simulation time and $\qavg$ when studying Eq.~\eqref{eq:prl-model} with a fixed set of simulation parameters and CPU hardware. We can produce an analogous plot for the ensemble of random rotations of this model, \cref{fig:app-prl-rot-timing-study}. Like the simple two qubit case, simulation time and $\qavg$ follow a dependence on $\lambda$ analogous to ES and FS. Yet, in addition, this ensemble can have a mild sign problem, which manifests particularly at $\lambda = -1$ (where the worse average sign is $0.85(4)$). This sign problem asymmetry at the two peaks results in a corresponding asymmetry in the simulation time. In other words, total simulation time---even for fixed simulation parameters and $\qavg$---can be adversely affected by the sign problem, as one would expect. 

\begin{figure*}[htp!]
    \centering
    \includegraphics[width=0.33\linewidth]{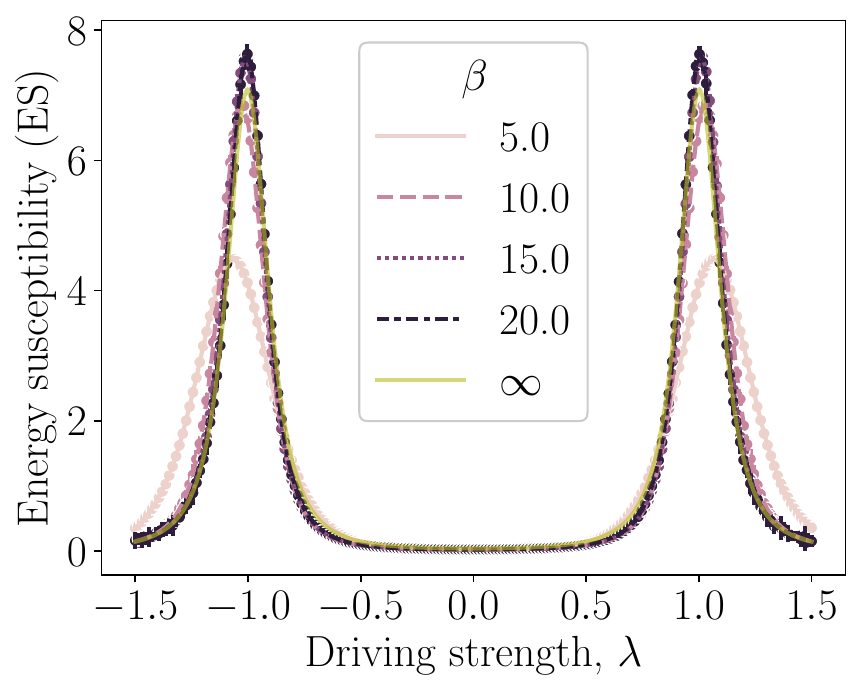}
    \includegraphics[width=0.33\linewidth]{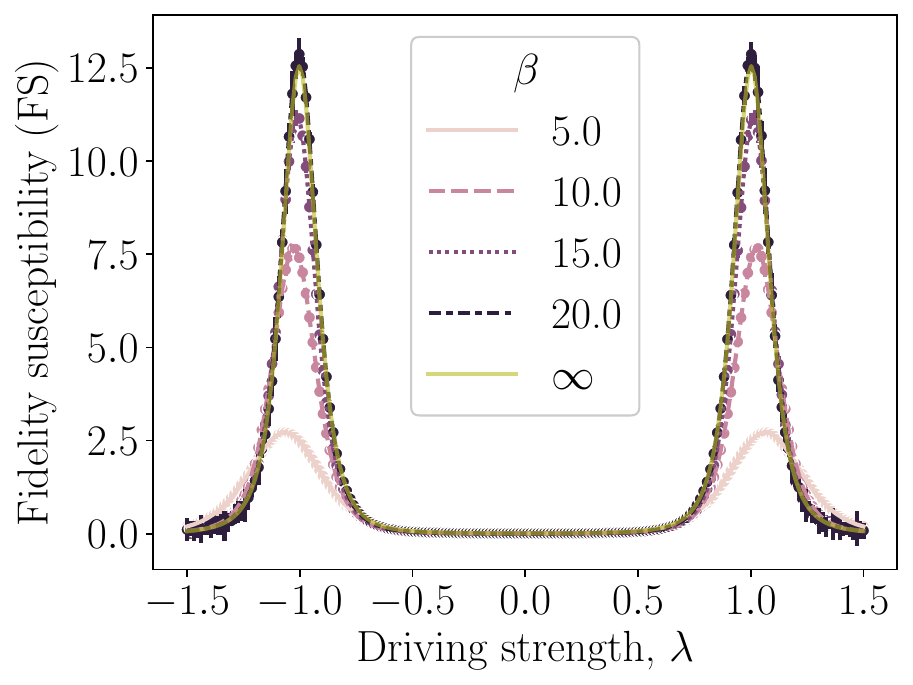}
    \includegraphics[width=0.33\linewidth]{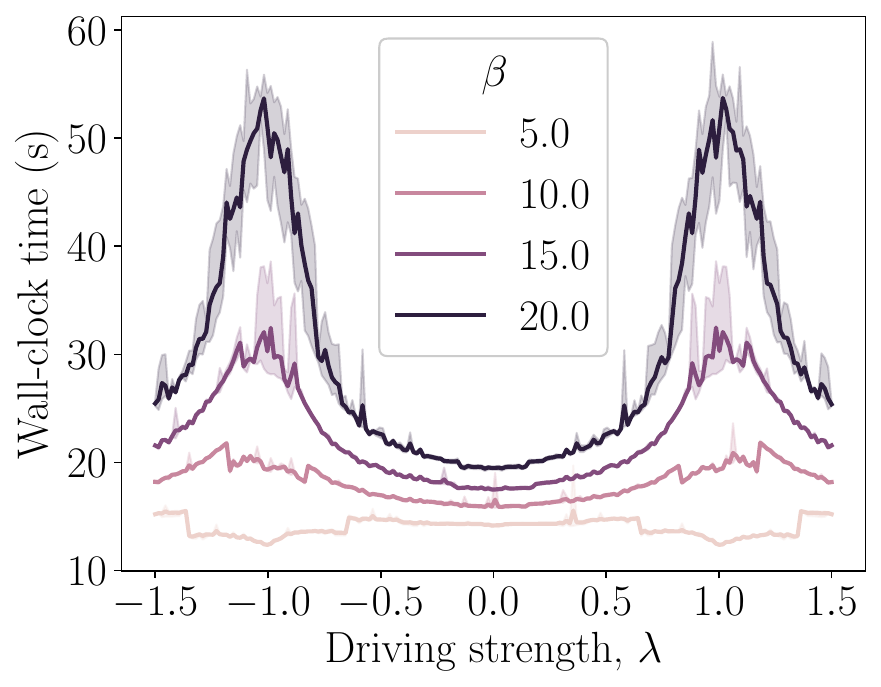}
    \includegraphics[width=0.33\linewidth]{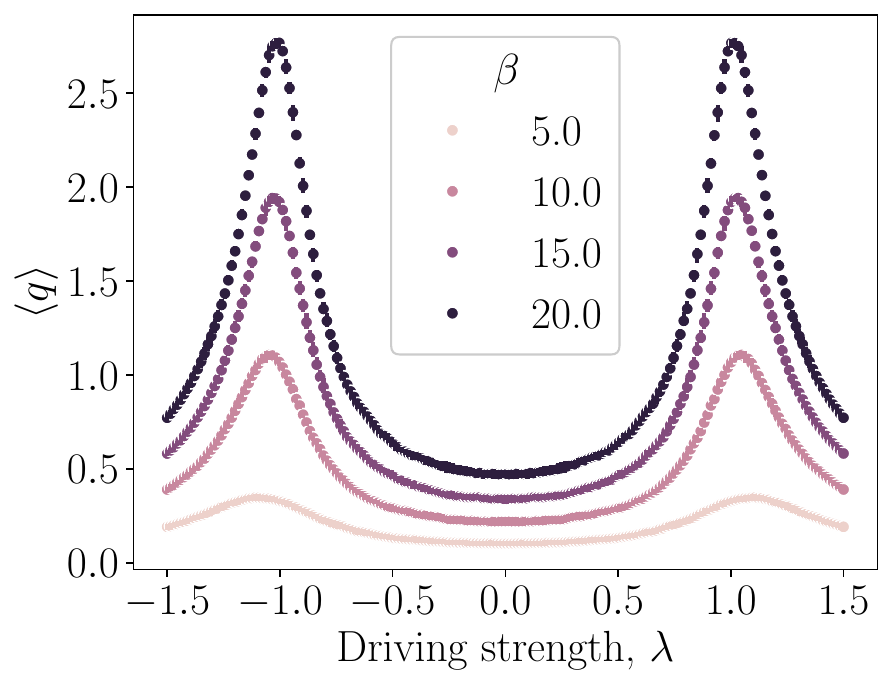}
    \caption{We again show the data plotted in Figs.~1 and 5 with the addition of an ES and $\qavg$ plot in studying Eq.~\eqref{eq:prl-model} Simulations were performed on 2.3 GHz Intel i9 CPU with input parameters as given in \cref{tab:app_simulation_parameters}.}
    \label{fig:app-prl-timing-study}
\end{figure*}

\begin{figure*}[htp!]
    \centering
    \includegraphics[width=0.32\linewidth]{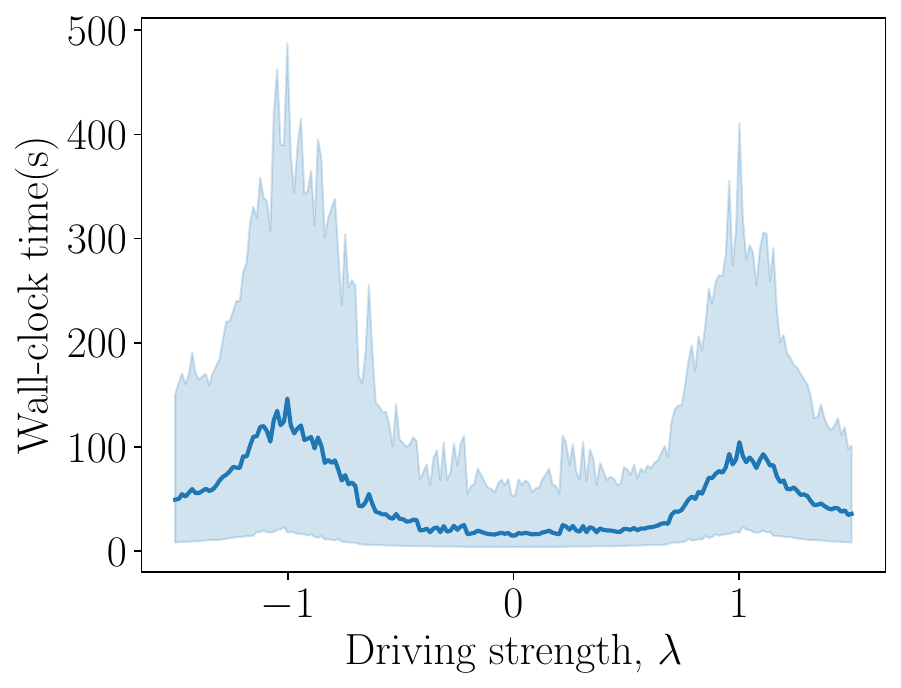}
    \includegraphics[width=0.32\linewidth]{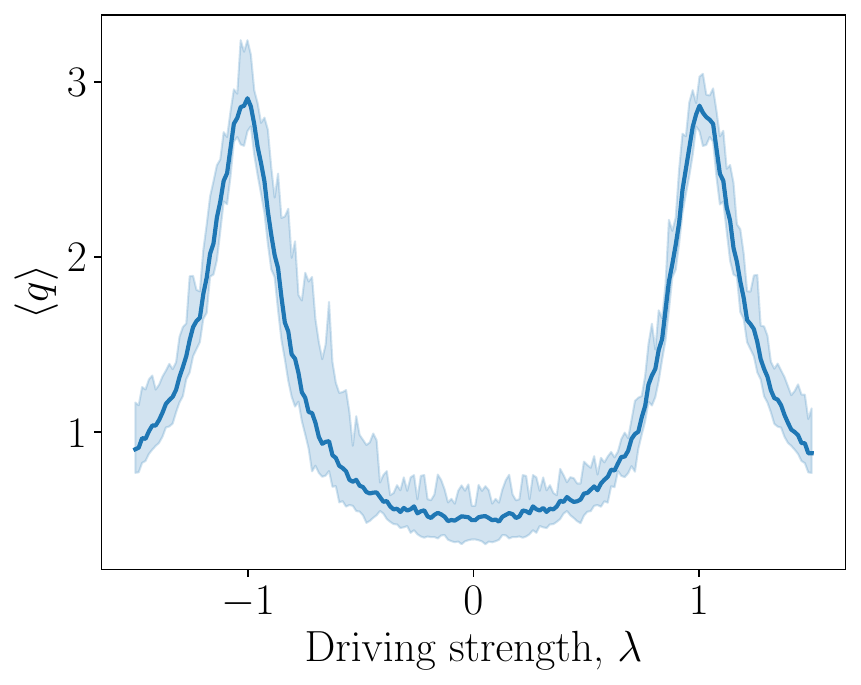}
    \includegraphics[width=0.32\linewidth]{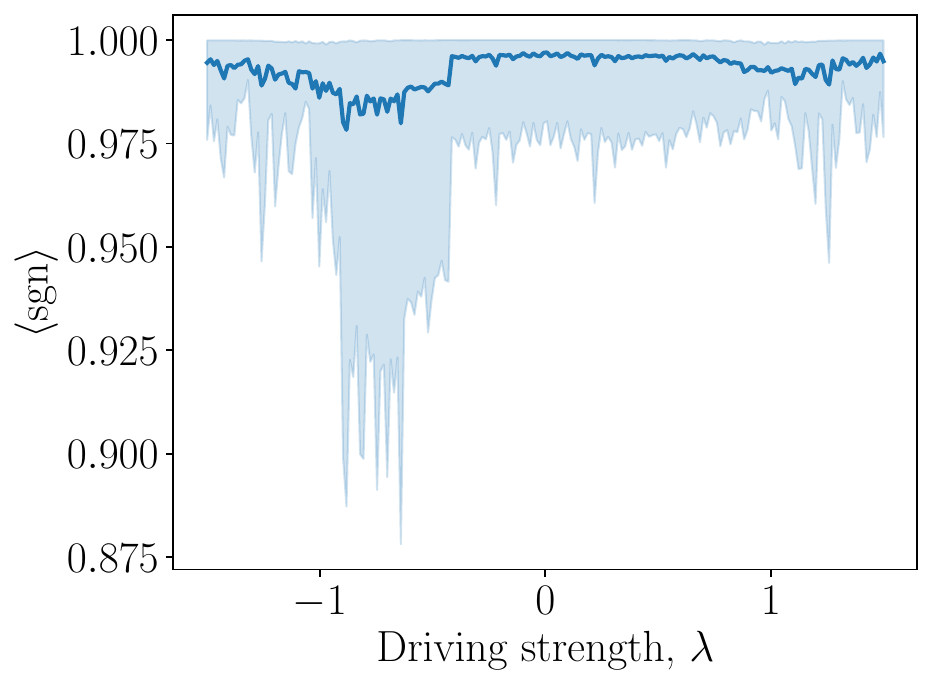}
    \caption{A timing study of the randomly rotated, post-selected ensemble studied in Fig.~4 and described in \cref{app:random-U}.}
    \label{fig:app-prl-rot-timing-study}
\end{figure*}

\begin{figure*}[htp!]
    \centering
    \includegraphics[width=0.32\linewidth]{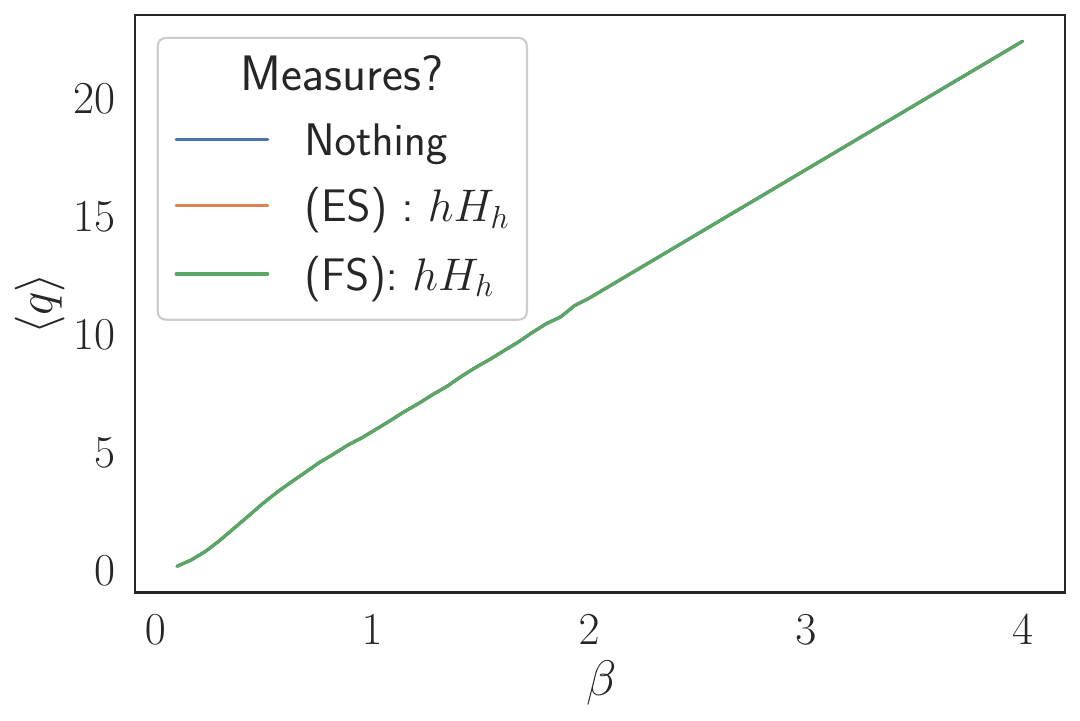}
    \includegraphics[width=0.32\linewidth]{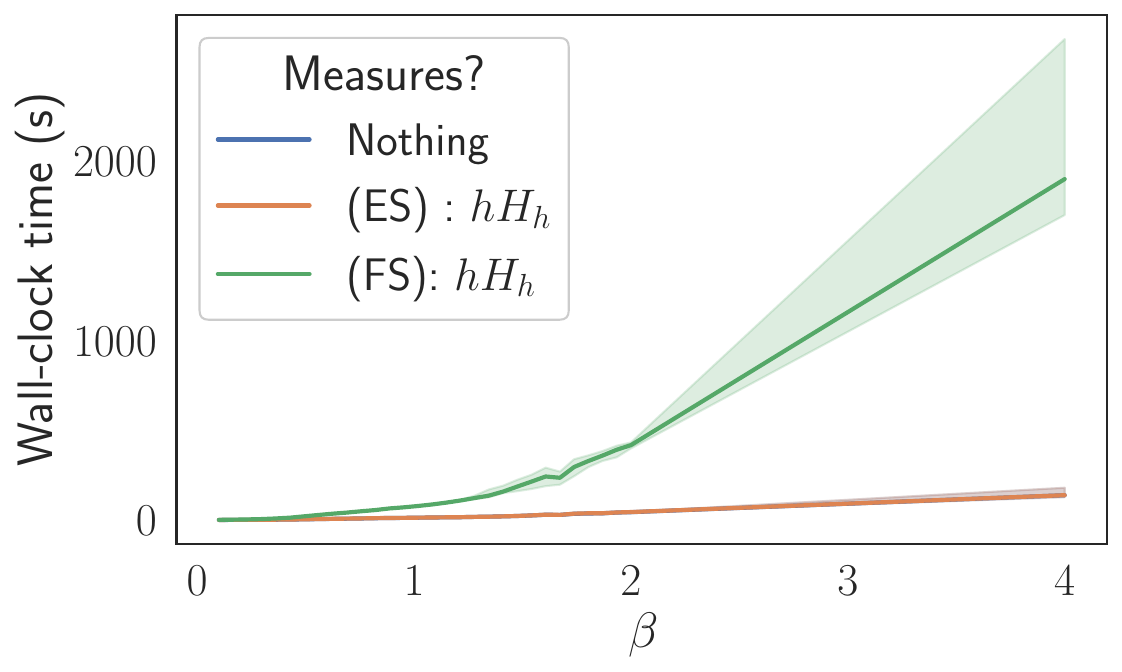}
    \includegraphics[width=0.32\linewidth]{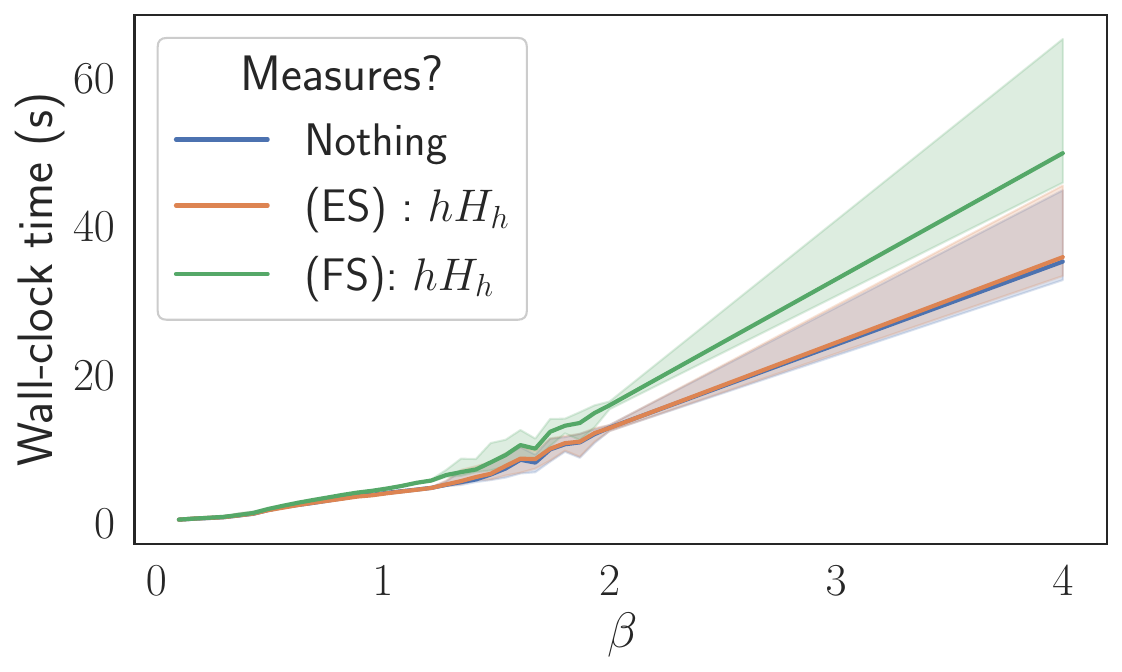}
    \caption{We re-simulate the $4 \times 4$ TFIM model studied in Fig.~2. Simulations were performed on $\nthreads = 6$ threads in parallel on 2.3 GHz Intel i9 CPUs with $\Tsteps = 0$, $\steps = 10^6$ and $\nbins = 100$. RNG seeds were fixed for each thread across different $\beta$. The three curves correspond to three measurement strategies: measure nothing, measure only (diagonal) ES, or measure only (diagonal) FS. (Left) We show $\qavg$ as a function of $\beta$. (b) We set $\stepsPer = 1$ so that we sample estimators densely at each QMC step/update. (c) We set $\stepsPer = 100$ so that we sample estimators sparsely every 100 QMC steps.}
    \label{fig:re-investigating-tfim-timing}
\end{figure*}

\begin{figure}
    \centering
    \includegraphics[width=0.5\linewidth]{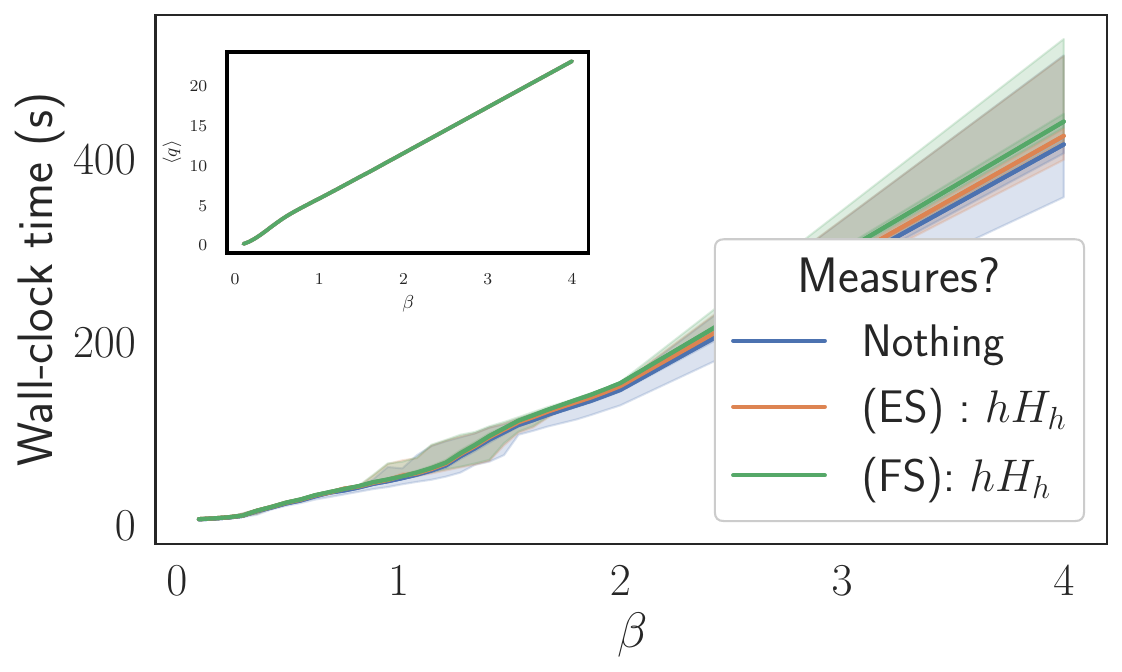}
    \caption{We repeat the numerical experiment on \cref{fig:re-investigating-tfim-timing} but with parameters that more closely match those actually used to generate Fig.~2. Namely, we choose $\Tsteps = 10^6$, $\steps = 10^7$, $\stepsPer = 10^3$, $\nbins = 100$, and $\nthreads = 50$. Using this many cores, we again use the USC HPC cluster but this time with a fixed 2.6\,GHz Intel Xeon E5-2640 CPU for all simulations. Wall-clock time scales as a modest power law in beta, $\sim \beta^{1.23}$ for all three curves, consistent with Fig. 3 of Ref.~\cite{barash2024QuantumMonteCarlo}.}
    \label{fig:closer-to-real-tfim-fig3-timing}
\end{figure}

More generally, the instantaneous complexity of QMC updates and measurements in our algorithm depends on the dynamic quantum dimension parameter $q$, and we know $\qavg$ itself scales with the difficulty of simulation. For example, when studying random quantum spin glasses with $N$ spins and a transverse field term of strength $\Gamma$, Ref.~\cite{albash2017OffdiagonalExpansionQuantum} observed an empirical scaling of $\qavg \sim \Gamma^2 \beta N.$ In \cref{fig:re-investigating-tfim-timing}, we re-simulate the $4 \times 4$ TFIM studied in Fig.~2, with general simulation parameters as listed in the caption. This plot empirically illustrates two things. First, in the left plot, we indeed find $\qavg \sim \beta$ as observed previously~\cite{albash2017OffdiagonalExpansionQuantum}. Second, the sparsity of measurements (controlled by $\stepsPer$) has a large effect on simulation time when estimating expensive observables such as FS. Furthermore, we verify empirically the claim that the ES has $O(1)$ cost, as simulation time when measuring ES or doing no measurement matches each other well regardless of $\stepsPer.$ As summarized in \cref{tab:app_simulation_parameters}, the actual simulation parameters used to generate Fig.~2 were in the sparse measurement regime of $\stepsPer = 10^3$. We verify this empirically in \cref{fig:closer-to-real-tfim-fig3-timing} by repeating the TFIM calculation with comparable parameters. Specifically, we find that simulation time is roughly independent of the measurement strategy, even when measuring the expensive FS quantity. 
}

\section{Parity restricted PMR-QMC for the TFIM}
\label{app:tfim-parity-details}
In our transverse-field Ising model (TFIM) numerical experiments of \cref{sec:numerics}, i.e. \cref{fig:prb-replicate-fig3}, we computed the ES and FS restricted to the positive parity subspace of the TFIM. We briefly discuss the details of our methodology. First, we study a rotated TFIM for which the transverse field term is $\sum_i Z_i$, Eq.~\eqref{eq:fidsus-tfim}. In doing so, the parity operator,
\begin{equation}
    \mathcal{P} = \prod_{i=1}^N Z_i
\end{equation}
commutes with the TFIM, $[H, \mathcal{P}] = 0$. Hence, the parity $P = \pm 1$ is a good quantum number, and it is well known that the groundstate of the TFIM lies in the $P = +1$ parity sector by the Perron-Frobenius theorem~\cite{albuquerque2010QuantumCriticalScaling}. 

For our rotated TFIM choice, $\mathcal{P}$ is purely diagonal. For both SSE-QMC~\cite{albuquerque2010QuantumCriticalScaling} and PMR-QMC (see Eq.~\cref{eq:pmr-z}), computational basis elements $\ket{z}$ are part of the QMC configuration. At any given time in simulation, then, one can query the eigenvalue of $\mathcal{P}  \ket{z}$, which is the instantaneous parity $P_{\mathcal{C}}$. Note that, in fact, $\ket{z}$ is a bitstring, so the parity can readily be obtained by determining if there are an odd number of $1$'s in $\ket{z}$ or not. In the SSE-QMC approach~\cite{albuquerque2010QuantumCriticalScaling}, the authors simply keep track of $P_{\mathcal{C}}$ during simulation and post-select on measurement outcomes. 

In contrast, we perform PMR-QMC sampling restricted to the positive parity subspace the entire simulation. This is possible because the set of PMR-QMC updates is standardized in terms of the permutation group that comprises $H$. A careful look at the updates for spin-1/2 systems~\cite{barash2024QuantumMonteCarlo} reveals that there are only two updates that can change a configurations parity: a classical update or a block swap. In our TFIM implementation~\cite{ezzell2025code}, we simply choose an initial $\ket{z}$ with positive parity and only allow classical updates or block swaps that preserve parity. Both can be altered to preserve parity without post-selection, which is therefore more efficient than the SSE approach. 

In a standard classical update, a random bit is flipped. In our code, we simply flip two random bits instead, which preserves parity. In a standard block swap, the classical state is also updated. Specifically, the permutation string $\Siq$ is split into two $\Siq = S_2 S_1$ for $S_1 = P_{i_k} \cdots P_{i_1}$ and $S_2 = P_{i_q} \cdots P_{i_{k+1}}$ for $k \in \{1, \ldots, q-1\}$  chosen randomly. The classical state at position $k$ is then $\ket{z'} = S_1 \ket{z}$ with energy $E_{z'}$ The updated block-swap configuration is then $\mathcal{C'} = \{\ket{z'}, S_1 S_2\}$, which swaps $S_1$ an $S_2$ while updating $\ket{z} \rightarrow \ket{z'}$. The parity $\ket{z'}$ can change based on $S_1$.

For the TFIM, however, each $P_{i_l}$ is just a 2-body $X$ string, i.e., $X_1 X_2$,  $\Siq = \mathds{1}$ if and only if each permutation appears an even number of times. Together, this means $S_1$ flips the parity of $\ket{z'}$ if and only if $k$ is odd. Hence, parity can be preserved by only sampling $k \in \{1, \ldots, q-1\}$ that is also even. This can easily be accomplished, and is what we do in our code~\cite{ezzell2025code}. Together, these small changes ensure we stay in the $P = +1$ parity sector during the entire PMR-QMC simulation.

\section{Additional details on generating special unitaries}
\label{app:random-U}

Consider $U = \sum_{k=1}^{L} c_k Q_k$ for $\cvec \in \mathds{R}^L$ and each $Q_k$ a Pauli matrix over $n$ spins (we use $Q$ to distinguish from the permutation $P_j$ in the PMR form of $H$). By construction, we can efficiently form the matrix $U H(\lambda) U^{\dagger}$ for $H(\lambda) = Z_1 Z_2 + 0.1 (X_1 + X_2) + \lambda(Z_1 + Z_2),$ our prototypical example Hamiltonian from Ref.~\cite{zhang2008detection}, whenever $L$ is polynomial in $n$. This is because for each term in $H$, we need only perform $L^2$ Pauli conjugations, i.e, compute terms like $Q_k X_1 Q_l$, which are themselves computable in $O(n)$ time using a look-up table or a symplectic inner product technique as in Clifford simulations.

In addition, $U$ can easily be shown to be unitary whenever $||\cvec||_2 = 1$ and $\{Q_k\}_{k=1}^L$ all anti-commute~(see Eq.~(6) in \cite{izmaylov2019unitary}). The largest maximal set of anti-commuting Paulis is $2n + 1$, and there is an $O(n)$ algorithm to generate a canonical set of this type~\cite{sarkar2019sets}. Let $\mathcal{Q}$ denote this canonical set of $201$ Paulis for a 100--spin system, which, in particular, includes $X^{\otimes 100}$. We can now describe our algorithm to generate 10 random unitaries such that $H_{U_i} \equiv U_i H(\lambda) U_i^{\dagger}$ satisfies the four properties described in the main text.

To begin with, we generate 50 random unitaries by (a) selecting $L \in [1, \ldots, 201]$ uniformly at random, (b) construct a random list of anti-commuting Paulis of the form $\{X^{\otimes 100}\} \cup \{Q_i\}_{i=2}^{L}$ for each $Q_i$ drawn uniformly without replacement from $\mathcal{Q}$, (c) generating $c_i \sim  \mathcal{N}(0, 1)$ for $i = 1, \ldots, L$ independently and then renormalizing. In step (b), we always include $X^{\otimes 100}$ to ensure each $H_{U_i}(\lambda) \equiv U_i H(\lambda) U^{\dagger}_i$ has non-trivial support on all 100 spins. By construction, $H_{U_i}$ has hundreds of Paulis terms, and yet, it is easy to generate and store. What remains to show is that we can find a subset of 10 without a severe sign problem.

\begin{figure}[htp!]
    \centering
    \includegraphics[width=1.0\linewidth]{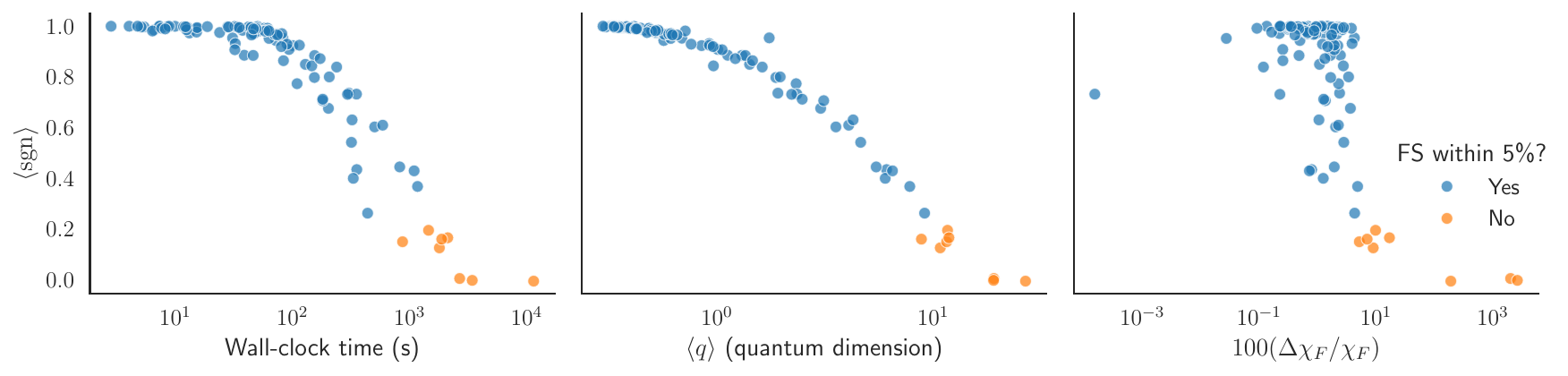}
    \caption{Here, we plot summary statistics of all 50 random models described generally above. Each plot shows how the average sign (for fixed QMC parameters) is correlated with important simulation properties like (left) simulation time (b) average value of $q$, and (c) percent error in estimating the FS. In Fig.~4, we only simulate those 10 random models with the highest average sign in this plot.}
    \label{fig:post-selection}
\end{figure}

To our knowledge, there is no simple calculation to check for the sign problem for these models a priori. So in practice, we simply ran a PMR-QMC simulation for each model for $\beta = 5$ and $\lambda = 1,$ which approximately the critical point for this model. The PMR-QMC code in Ref.~\cite{barash2024PmrQmcCode, barash2024QuantumMonteCarlo} automatically tracks the average and variance of the sign of the QMC weight, so we then simply post-selected those models with the highest average sign. Empirically, we find the worst average sign in this set of 10 is $0.9983 \pm  0.0006$ (this is $2\sigma$), so this subset is essentially sign problem free at $\lambda = 1.0$ { (but this does not hold for all $\lambda$ as seen before in \cref{fig:app-prl-rot-timing-study}). This post-selection step is contained in a Jupyter notebook of our code~\cite{ezzell2025code}, namely in \verb^before_fig2_post_selection.ipynb^. We provide a graphical summary of the practical meaning of this post selection process in \cref{fig:post-selection}. In short, the worse the sign problem is, the longer the simulation takes to run and the less accurate it is for a fixed set of QMC simulation parameters. By post-selecting on those models with the 10 highest average signs, we thus generate an ensemble with random, high weight Pauli terms that is nevertheless efficiently simulable in PMR-QMC.  
} 

\section{Derivations of PMR-QMC estimators for pure diagonal or pure off-diagonal driving}
\label{app:best-estimator-derivations}

We derive the ES and FS estimators for the $H_1 \propto D_0$ case stated in \cref{sec:methods}.  Most notably, we prove \cref{thrm:constant-ES-estimator,thrm:cubic-FS-estimator}.   We then substantiate the claim that these same estimators also straightforwardly carry over to the $H_1 \propto H - D_0$ case.

\subsection{Derivation strategy}

Recall from the main text (or see Ref.~\cite{ezzell2025advanced}) that a PMR-QMC estimator  $O_{\mathcal{C}}$ is a function of $\mathcal{C}$ such that
\begin{equation}
    \label{app-eq:pmr-qmc-estimator}
    \avg{O} = \frac{\sum_{\mathcal{C}} w_{\mathcal{C}} O_{\mathcal{C}}}{\sum_{\mathcal{C}} w_{\mathcal{C}}},  
\end{equation}
for $w_{\mathcal{C}} = D_{(z, \Siq)} e^{-\beta [E_{z_0}, \ldots, E_{z_q}]}$ the generalized PMR-QMC Boltzmann weight. For convenience, we write $O_{\mathcal{C}} \estimates \avg{O}$ to be read ``$O_{\mathcal{C}}$ estimates $\avg{O}$." Our first goal is to derive an estimator for the ES when $H_1 \propto D_0$. 

In most practical applications, $H_1 \propto D_0$ follows from the relation $\lambda H_1 = D_0$, i.e., the driving term multiplied by $\lambda$ is the diagonal portion of $H$. As such, $\avg{H_1} = \avg{D_0} / \lambda$, and more generally, it is enough to simply estimate $\avg{D_0}$ and related quantities. To this end, we observe
\begin{equation}
    \label{app-eq:d0-es}
    \chi_E^{D_0}(\lambda; \beta) = \int_0^\beta (\avg{D_0(\tau) D_0} - \avg{D_0}^2) \dtau.
\end{equation}
gives the ES (see Eq.~\cref{eq:gibbs-ES}) up to a multiplicative correction, i.e., $\chi_E^{H_1} = \chi_E^{D_0} / \lambda^2$ (or more generally just $\chi_E^{H_1} \propto \chi_E^{D_0})$. Thus, we focus on attention on estimating $\chi_E^{D_0}$. 

Our strategy is to first derive estimators for $\avg{D_0}$ and $\avg{D_0(\tau) D_0}$ in \cref{appsub:d0-correlator}. As we shall see, $E_z \estimates \avg{D_0}$ does not depend on $\tau$. Then, in \cref{appsub:es-in-quadratic-time}, we show $\int_0^\beta \avg{D_0(\tau) D_0}$ can be evaluated exactly and analytically into a valid estimator via novel divided difference relations. Together, these results give an $O(q^2)$ estimator for ES. In \cref{appsub:es-in-quadratic-time}, we further simplify our estimator into a constant time $O(1)$ ES estimator using additional novel divided difference relations. This completes most of the work of proving Theorem 1.

By a similar chain of reasoning, 
\begin{equation}
    \label{app-eq:d0-fs}
    \chi_F^{D_0}(\lambda; \beta) = \int_0^{\beta/2} \tau (\avg{D_0(\tau) D_0} - \avg{D_0}^2) \dtau,
\end{equation}
gives the FS (see Eq.~\eqref{eq:gibbs-FS}) up to a multiplicative correction, and our strategy remains similar. Namely, we show in \cref{app-sub:fs-in-quartic} that $\int_0^{\beta/2} \tau \avg{D_0(\tau) D_0} \dtau$ can also be evaluated exactly and analytically via novel divided difference relations. This results in an $O(q^4)$ estimator for FS. By yet additional novel divided difference relations, we reduce this to $O(q^3)$ operations in \cref{app-sub:fs-in-cubic}. This completes most of the work of proving Theorem 2. 

Finally, we show how all the estimators we derived for the $H \propto D_0$ case carry over straightforwardly to the $H_1 \propto H - D_0$ case and with the same complexities in \cref{app-sub:off-diagonal-case-from-diagonal} thus completing the justification for Theorems 1 and 2. 

\subsection{\texorpdfstring{Estimators for $\avg{D_0}$ and $\avg{D_0(\tau)\, D_0}$}{Estimators for <D0> and <D0(tau) D0>}}
\label{appsub:d0-correlator}

Since $\avg{z | D_0 e^{-\beta H} | z} = E_z \avg{z | e^{-\beta H} | z}$, then by a direct off-diagonal series expansion of $\Tr[D_0 e^{-\beta H}]$ (see also Sec.~VIII.A of Ref.~\cite{ezzell2025advanced}), we find
\begin{equation}
    \label{app-eq:d0-estimator}
    E_z \estimates \avg{D_0},
\end{equation}
is an $O(1)$ estimator (see \cref{app-sub:complexity-in-pmr-qmc}). 

Next, we derive an estimator for $\avg{D_0(\tau) D_0}$. One approach is by direct off-diagonal series expansion and simplification~\cite{ezzell2025advanced}.  Here, we argue that direct Leibniz rule logic (\cref{eq:leibniz-rule}),  we can correctly conclude,  
\begin{equation}
	\avg{z | D_0 e^{- \tau H}  D_0 e^{-(\beta - \tau) H} | z} = E_z \sum_{j=0}^q E_{z_j} e^{-\tau [E_{z_j}, \ldots, E_{z_q}]} e^{-(\beta - \tau)[E_{z_0}, \ldots, E_{z_j}]}.
\end{equation}
The basic idea is simple: without the middle $D_0$ term,  we can safely apply the Leibniz rule to the off-diagonal expansion since $e^{-\tau H}$ and $e^{-(\beta - \tau) H}$ commute.  Accounting for the non-commuting $D_0$,  however,  is simple---we just pick up a factor of $E_{z_j}$ from the  $D_0 \ket{z_j}$  as shown.  Hence,  we find 
\begin{equation}
    \label{app-eq:d0-correlator-estimator}
    (D_0(\tau)D_0)_{\mathcal{C}} \equiv \\ \frac{E_{z}}{e^{-\beta[E_{z_0}, \ldots, E_{z_q}]}} \sum_{j=0}^q E_{z_j} e^{-\tau [E_{z_j}, \ldots, E_{z_q}]} e^{-(\beta - \tau)[E_{z_0}, \ldots, E_{z_j}]} \estimates \avg{D_0(\tau) D_0}.
\end{equation}
as claimed in Eq.~\eqref{eq:estimator-D0-corr}. 

\subsection{Energy susceptibility in quadratic time}
\label{appsub:es-in-quadratic-time}

We shall show
\begin{equation}
    \label{app-eq:ES-estimator-orig}
    \frac{- E_{z_0}}{e^{-\beta[E_{z_0}, \ldots, E_{z_q}]}} \sum_{j=0}^q E_{z_j} e^{-\beta[E_{z_0}, \ldots, E_{z_q}, E_{z_j}]} \estimates \int_0^\beta \avg{ D_0(\tau) D_0 } \dtau.
\end{equation}
Combined with \cref{app-eq:d0-estimator,app-eq:d0-es}, this gives an estimator for $\chi_E^{D_0}$ that requires $O(q^2)$ operations (see \cref{app-sub:complexity-in-pmr-qmc}). By inspection of \cref{app-eq:d0-correlator-estimator}, showing \cref{app-eq:ES-estimator-orig} amounts to proving the following integral relation. 

\begin{lemma}[Convolution lemma~\cite{zeng2025inequalities}]
    \label{app-lem:convolution}
    \begin{equation}
    \label{app-eq:convolution-theorem}
    \int_0^\beta e^{-\tau [x_{j+1}, \ldots, x_q]}  e^{-(\beta-\tau) [x_0, \ldots, x_j]} \dtau = - e^{-\beta [x_0, \ldots, x_q]}.
\end{equation}
\end{lemma}
\begin{proof}
    For convenience, we define the functions
\begin{align}
    f(t) &= e^{-t [x_{j+1}, \ldots, x_q]} \\
    g(t) &= e^{-t [x_0, \ldots, x_j]}
\end{align}
The convolution of these functions,
\begin{align}
    (f * g)(t) &= \int_0^t f(\tau) g(t - \tau) \dtau 
\end{align}
is by construction the integral we want to evaluate for $t = \beta$. Let $\lap{f(t)}$ denote the Laplace transform of $f(t)$ from $t \rightarrow s$. By the convolution property of the Laplace transform and~\cref{eq:laplace-of-dd}, we find
\begin{align}
    \lap{(f * g)(t)} &= \lap{f(t)} \lap{g(t)} \\
    &= \left( \frac{ (-1)^{q-j-1} }{ \prod_{l=j+1}^q (s + x_l) } \right) \left( \frac{ (-1)^{j} }{ \prod_{m=0}^j (s + x_m) }\right) = \frac{ (-1)^{q-1} }{ \prod_{l=0}^q (s + x_l) } = \lap{ - e^{-t [x_0, \ldots, x_q]} }.
\end{align}
Taking the inverse Laplace transform of the first and final expression proves the claimed integral relation.
\end{proof} 

Inspection of \cref{app-eq:d0-correlator-estimator} and \cref{app-eq:convolution-theorem} shows the validity of \cref{app-eq:ES-estimator-orig} immediately. As an aside, we remark that a direct proof by series expanding both DDE via~\cref{eq:dd-as-series-expansion}, integrating term-by-term, regrouping, and simplifying is also possible.

\subsection{Energy susceptibility in constant time}
\label{appsub:es-in-constant-time}

The ES estimator in \cref{app-eq:ES-estimator-orig} requires $O(q^2)$ operations to evaluate, but it can be simplified into an $O(1)$ estimator and validate Theorem 1 in a few steps. We begin with a simple proposition. 

\begin{proposition}[Re-scaling relation~\cite{zeng2025inequalities}]
    \label{prop:rescaling}
    \begin{equation}
        e^{-\tau[x_0, \ldots, x_q]} = (-\tau)^q e^{[-\tau x_0, \ldots, -\tau x_q]}.
    \end{equation}
\end{proposition}
\begin{proof}
    By the change of variables $y = -\tau z$ in Eq.~\eqref{eq:contour-int-dd}, we find the claimed result. 
\end{proof}

We can now simplify the sum in \cref{app-eq:ES-estimator-orig} into a compact derivative form. 

\begin{lemma}[A weighted, repeated argument sum simplification~\cite{zeng2025inequalities}]
\label{lemm:weighted-repeatedarg-sum}
    \beq
        \label{eq:weighted-repeatedarg-sum}
        \sum_{j=0}^q x_j e^{-\tau [x_0,\dots,x_q,x_j]} = \left(\tau\frac{\partial}{\partial\tau}-q\right)e^{-\tau [x_0,\dots,x_q]}.
    \eeq
\end{lemma}

\begin{proof}
We prove this by induction.

\emph{Base case.}
When $q=0$, we have 
$x_0 e^{-\tau [x_0,x_0]} = x_0 \partial e^{-\tau x_0} / \partial x_0 = -\tau x_0 e^{-\tau x_0}$, so Eq.~(\ref{eq:weighted-repeatedarg-sum}) holds.

\emph{Induction step.} By Prop.~\ref{prop:rescaling}, we write the right-hand side as,
\begin{align}
    (\tau \partial_\tau - (q + 1)) (-\tau)^{q+1} e^{[-\tau x_0, \ldots, -\tau x_{q+1}]} &= -(-\tau)^{q+2} \pdv{\tau} e^{[-\tau x_0, \ldots, -\tau x_{q+1}]} \\
    &= -(-\tau)^{q+2} \frac{1}{2\pi i} \oint_{\Gamma} \pdv{\tau} \frac{e^z}{\prod_{i=0}^{q+1} (z + \tau x_i)} \dz = \sum_{j=0}^{q+1} x_j e^{-\tau [x_0, \ldots, x_{q+1}, x_j]}
\end{align}

Therefore, Eq.~(\ref{eq:weighted-repeatedarg-sum}) is true for all $q$.  
\end{proof}

To turn this simplification into a useful result for PMR-QMC estimation, we employ the following proposition. 

\begin{proposition}[Parametric derivative]~\cite{zeng2025inequalities}
    \label{prop:parametric-derivative-1}
    \begin{equation}
        \pdv{\tau} e^{-\tau [x_0, \ldots, x_q]} = - x_0 e^{-\tau [x_0, \ldots, x_q]} - e^{-\tau [x_1, \ldots, x_q]}.
    \end{equation}
\end{proposition}
\begin{proof}
    First, define $g(x) = -x e^{-\tau x}$ for convenience. By Eq.~\eqref{eq:contour-int-dd}, 
    \begin{align}
        \pdv{\tau} e^{-\tau [x_0, \ldots, x_q]} \equiv \pdv{\tau} \frac{1}{2\pi i} \oint_{\Gamma} \frac{e^{-\tau z}}{\prod_{i=0}^q (z - x_i)} \dz =  \frac{1}{2\pi i} \oint_{\Gamma} \frac{-z e^{-\tau z}}{\prod_{i=0}^q (z - x_i)} \dz \equiv g[x_0, \ldots, x_q].
    \end{align}
    Employing the Leibniz rule, \cref{eq:leibniz-rule}, we get the desired result. 
\end{proof}

\begin{corollary}[Computing the weighted, repeated argument sum~\cite{zeng2025inequalities}]
\beq
    \label{eq:computing-weighted-repeatedarg-sum}
    \sum_{j=0}^q x_j e^{-\tau [x_0,\dots,x_q,x_j]} = 
    	\left\{\begin{array}{ll}
            (-x_0\tau - q)e^{-\tau[x_0,\dots,x_q]}-\tau e^{-\tau[x_1,\dots,x_q]}, & \text{\rm for}\quad q > 0,\\
            -\tau x_0 e^{-\tau x_0}, & \text{\rm for}\quad q = 0.
            \end{array}\right.
\eeq
\end{corollary}
\begin{proof}
    This follows from applying Prop.~\ref{prop:parametric-derivative-1} to Lem.~\ref{lemm:weighted-repeatedarg-sum}. 
\end{proof}

We now prove the main part of Theorem 1---that we have an $O(1)$ estimator for ES---since 
\beq
        \beta E_{z_0}^2 + q E_{z_0} + \mathbf{1}_{q>0} \beta E_{z_0} \frac{ e^{-\beta [E_{z_0}, \ldots, E_{z_{q-1}}] } }{ e^{-\beta [E_{z_0}, \ldots, E_{z_q}]} } \estimates \int_0^\beta \avg{ D_0(\tau) D_0 } \dtau.
    \eeq
\begin{proof}
    The original ES estimator in Eq.~\eqref{app-eq:ES-estimator-orig} can be simplified using  Eq.~\eqref{eq:computing-weighted-repeatedarg-sum}. The first two terms follow by simple algebraic manipulation. The final term follows by algebra and periodicity, i.e., $E_{z_0} = E_{z_q}$ for a valid set of PMR-QMC diagonal energies. After a PMR-QMC update, we have $O(1)$ access to $\beta$, each $E_{z_j},$ and each successive divided difference $e^{-\beta E_{z_0}}, e^{-\beta [E_{z_0}, E_{z_1}]}, \ldots, e^{-\beta [E_{z_0}, \ldots, E_{z_q}]}$, which includes $e^{-\beta [E_{z_0}, \ldots, E_{z_{q-1}}]}$.
\end{proof}

\subsection{Fidelity susceptibility in quartic time}
\label{app-sub:fs-in-quartic}

We shall now show
\beq
    \label{app-eq:FS-estimator-orig}
      \frac{E_{z_0}}{e^{-\beta[E_{z_0}, \ldots, E_{z_q}]}}\sum_{j=0}^q E_{z_j} 
    \sum_{r=0}^j e^{-\frac{\beta}{2}[E_{z_0}, \ldots, E_{z_r}]} \sum_{m=j}^q e^{-\frac{\beta}{2}[E_{z_r}, \ldots, E_{z_q}, E_{z_j}, E_{z_m}]} \estimates \int_0^{\beta/2} \tau \avg{ D_0(\tau) D_0 } \dtau , 
\eeq
which gives rise to an $O(q^4)$ estimator (see \cref{app-sub:complexity-in-pmr-qmc}) for $\chi_F^{D_0}$ in \cref{app-eq:d0-fs} when combined with \cref{app-eq:d0-estimator}. As with the ES, this essentially reduces to proving divided difference relations since we need only integrate the expression in \cref{app-eq:d0-correlator-estimator}. To this end, we first state and prove a helpful proposition.

\begin{proposition}[Repeated argument sum simplification~\cite{zeng2025inequalities}]
    \label{prop:repatedarg-sum-simp}
    \begin{equation}
        \sum_{j=0}^q e^{-\tau [x_0, \ldots, x_q, x_j]} = -\tau e^{-\tau [x_0, \ldots, x_q]}
    \end{equation}
\end{proposition}
\begin{proof}
    Denote $f(\tau) \equiv e^{-\tau [x_0, \ldots, x_q]}$ and recall  $\lap{e^{-\tau [x_0, \ldots, x_q]}} = (-1)^q / (\prod_{i=0}^q (s + x_i)).$ By the derivative property of the Laplace transform,
        \begin{equation}
            \lap{-\tau f(\tau)} = \pdv{s} \lap{f(\tau)} =  \sum_{j=0}^q \frac{(-1)^{q+1}}{(s + x_j) \prod_{i=0}^q (s + x_i)} = \lap{\sum_{j=0}^q e^{-\tau [x_0, \ldots, x_q, x_j]}},
        \end{equation}
    and the result follows by taking the inverse Laplace transform of the left-most and right-most expressions. 
\end{proof}
Using this proposition and results we have shown previously, we can now prove the following integral relation. 
\begin{lemma}[Fidelity susceptibility integral]
    \label{lemm:fidsus-integral}
    \begin{equation}
    \label{app-eq:fs-integral}
    I \equiv \int_0^{\beta/2} \tau e^{-\tau [x_{j+1}, \ldots, x_q]}  e^{-(\beta-\tau) [x_0, \ldots, x_j]} \dtau = \sum_{r=0}^j e^{-\frac{\beta}{2}[x_0, \ldots, x_r]} \sum_{m = j + 1}^q e^{-\frac{\beta}{2} [x_r, \ldots, x_q, x_m]},
\end{equation}
\end{lemma}
\begin{proof}
    This relation follows from the Leibniz rule (\cref{eq:leibniz-rule}), ~\cref{prop:repatedarg-sum-simp}, and ~\cref{app-eq:convolution-theorem} applied in order,
\begin{align}
    I &= \sum_{r=0}^j e^{-(\beta/2 - \tau) [x_0, \ldots, x_r]} \int_0^{\beta/2} \tau e^{-\tau [x_{j+1}, \ldots, x_q]}  e^{-(\beta/2 - \tau) [x_r, \ldots, x_j]}\dtau \\    
    &= \sum_{r=0}^j \sum_{m=j+1}^q e^{-(\beta/2 - \tau) [x_0, \ldots, x_r]} \int_0^{\beta/2} e^{-\tau [x_{j+1}, \ldots, x_q, x_m]} e^{-(\beta/2 - \tau) [x_r, \ldots, x_j]} \dtau \\
     &= \sum_{r=0}^j  e^{-(\beta/2 - \tau) [x_0, \ldots, x_r]} \sum_{m=j+1}^q e^{-\frac{\beta}{2} [x_r, \ldots, x_q, x_m] },
\end{align}
where in the final line we use $\beta \rightarrow \beta /2$ when employing~\cref{app-eq:convolution-theorem}.
\end{proof}

This lemma is already sufficient to demonstrate \cref{app-eq:FS-estimator-orig}, but to make things clearer, we state the following corollary first. 
\begin{corollary}
    \cref{lemm:fidsus-integral} is general for any multiset $[x_0, \ldots, x_q]$ partitioned at arbitrary index $j$. As a specific special case of interest now, we consider $[x_{j+1}, \ldots, x_q]$ replaced with $[x_j, \ldots, x_q]$, 
\begin{equation}
    \label{app-eq:alt-fidsus-int}
    I' \equiv \int_0^{\beta/2} \tau e^{-\tau [x_{j}, \ldots, x_q]}  e^{-(\beta-\tau) [x_0, \ldots, x_j]} \dtau = \sum_{r=0}^j e^{-\frac{\beta}{2}[x_0, \ldots, x_r]} \sum_{m = j}^q e^{-\frac{\beta}{2} [x_r, \ldots, x_q, x_j, x_m]}.
\end{equation}
\end{corollary}
\begin{proof}
    This result follows from \cref{lemm:fidsus-integral} with the main change being that the $m$ index now starts at $j$ instead of $j+1$. To explain the notation change, it is useful to simply repeat the above derivation, which proceeds as,
\begin{align}
    I' &= \sum_{r=0}^j e^{-(\beta/2 - \tau) [x_0, \ldots, x_r]} \int_0^{\beta/2} \tau e^{-\tau [x_{j}, \ldots, x_q]}  e^{-(\beta/2 - \tau) [x_r, \ldots, x_j]}\dtau \\    
    &= \sum_{r=0}^j \sum_{m=j}^q e^{-(\beta/2 - \tau) [x_0, \ldots, x_r]} \int_0^{\beta/2} e^{-\tau [x_{j}, \ldots, x_q, x_m]} e^{-(\beta/2 - \tau) [x_r, \ldots, x_j]} \dtau \\
     &= \sum_{r=0}^j  e^{-(\beta/2 - \tau) [x_0, \ldots, x_r]} \sum_{m=j+1}^q e^{-\frac{\beta}{2} [x_r, \ldots, x_q, x_j, x_m] },
\end{align}
where the last multiset now contains an explicit extra $x_j$ from the concatenation of $[x_j, \ldots, x_q, x_m]$ and $[x_r, \ldots, x_j]$ which both contain an explicit $x_j$. 
\end{proof}
Inspection of \cref{app-eq:d0-correlator-estimator,app-eq:alt-fidsus-int} shows the validity of \cref{app-eq:FS-estimator-orig} immediately. 

\subsection{Fidelity susceptibility in cubic time}
\label{app-sub:fs-in-cubic}

The FS estimator in \cref{app-eq:FS-estimator-orig} can be evaluated in $O(q^4)$ operations, but we can improve it to $O(q^3)$ and validate Theorem 2 in a few steps. We begin with a some helpful propositions and lemmas, where the first step is to re-order the sum. 

\begin{proposition}[Re-ordering sums in FS estimator]
    \beq
    \label{eq:FS-estimator-orig-reorder}
    \frac{E_{z_0}}{e^{-\beta[E_{z_0}, \ldots, E_{z_q}]}}
    \sum_{r=0}^q e^{-\frac{\beta}{2}[E_{z_0}, \ldots, E_{z_r}]}
    \sum_{\substack{j,m=r \\ j \leq m}}^q
    E_{z_j} e^{-\frac{\beta}{2}[E_{z_r}, \ldots, E_{z_q}, E_{z_j}, E_{z_m}]} \estimates \int_0^{\beta/2} \tau \avg{ D_0(\tau) D_0 } \dtau 
    \eeq
\end{proposition}
With the re-ordering, we now observe that the innermost sum can be simplified. 
\begin{lemma}[A weighted, repeated argument, double, less-than sum simplification~\cite{zeng2025inequalities}]
    \label{lemm:weighted-repeated-double-lessthan-sum-simplification}
    \beq
        \label{eq:weighted-repeated-double-lessthan-sum-simplification}
        \sum_{\substack{i,j=0 \\ i \leq j}}^q x_i e^{-\tau [x_0,\dots,x_q,x_i,x_j]} = 
        \left(-\frac{\tau^2}{2}\frac{\partial}{\partial\tau}-\sum_{i=0}^q i\cdot \frac{\partial}{\partial x_i}\right) e^{-\tau [x_0,\dots,x_q]}.
    \eeq
\end{lemma}
\begin{proof}
We prove this by induction.

\emph{Base case}.
When $q=0$, we have 
$$
x_0 e^{-\tau [x_0,x_0,x_0]} = x_0 \cdot \frac{1}{2}\frac{\partial^2}{\partial x_0^2} e^{-\tau x_0} 
= \frac{\tau^2 x_0}{2} e^{-\tau x_0} = -\frac{\tau^2}{2}\frac{\partial}{\partial\tau} e^{-\tau x_0},
$$ so Eq.~(\ref{eq:weighted-repeated-double-lessthan-sum-simplification}) holds.

\emph{Induction step.} It follows from the recursive definition of the divided differences that
\begin{multline}
\sum_{\substack{i,j=0 \\ i \leq j}}^{q+1} x_i e^{-\tau [x_0,\dots,x_{q+1},x_i,x_j]} =
\frac{1}{x_0-x_1} \sum_{\substack{i,j=0 \\ i \leq j}}^{q+1} x_j \left( e^{-\tau [x_0,x_2,\dots,x_{q+1},x_i,x_j]} - e^{-\tau[x_1,x_2,\dots,x_{q+1},x_i,x_j]}\right) = \\ = 
\frac{1}{x_0-x_1} \left(
\left(
\left(-\frac{\tau^2}{2}\frac{\partial}{\partial\tau}-\sum_{i=1}^q i\cdot \frac{\partial}{\partial x_{i+1}}\right)e^{-\tau [x_0,x_2,\dots,x_{q+1}]} + 
x_0 e^{-\tau [x_0,x_2,\dots,x_{q+1},x_0,x_1]} + 
x_1 \sum_{j=1}^{q+1} e^{-\tau [x_0,\dots,x_{q+1},x_j]}
\right) -
\right. \\ \left.
-\left(
\left(-\frac{\tau^2}{2}\frac{\partial}{\partial\tau}-\sum_{i=1}^q i\cdot \frac{\partial}{\partial x_{i+1}}\right)e^{-\tau [x_1,x_2,\dots,x_{q+1}]} + x_0 \sum_{j=0}^{q+1} e^{-\tau [x_0,x_1,\dots, x_{q+1},x_j]}
\right)
\right) 
= \\ =
\left(-\frac{\tau^2}{2}\frac{\partial}{\partial\tau}-\sum_{i=0}^q i\cdot \frac{\partial}{\partial x_{i+1}}\right)e^{-\tau [x_0,\dots,x_{q+1}]}
-\sum_{j=1}^{q+1} e^{-\tau [x_0,\dots,x_{q+1},x_j]}
= \\ =
\left(-\frac{\tau^2}{2}\frac{\partial}{\partial\tau} -\sum_{i=0}^{q+1} i\cdot \frac{\partial}{\partial x_i}\right) e^{-\tau [x_0,\dots,x_{q+1}]}.
\end{multline}

Therefore, Eq.~(\ref{eq:weighted-repeated-double-lessthan-sum-simplification}) is true for all $q$.  
\end{proof}
As with the ES simplification, this more compact expression involving a derivative can be converted into a valid PMR-QMC estimator using \cref{prop:parametric-derivative-1}. 
\begin{corollary}[Computing the complicated weighted, repeated arugment, double, less-than sum~\cite{zeng2025inequalities}]
    \beq
        \label{eq:computing-weighted-repeated-double-lessthan-sum}
        \sum_{\substack{i,j=0 \\ i \leq j}}^q x_i e^{-\tau [x_0,\dots,x_q,x_i,x_j]} = 
        	\left\{\begin{array}{ll}
        	\left(\tau^2 x_0/2\right) e^{-\tau [x_0,\dots,x_q]} + (\tau^2/2) e^{-\tau [x_1,\dots,x_q]}
        	- \sum_{i=1}^q i \cdot e^{-\tau [x_0,\dots,x_q,x_i]}, & \text{\rm for}\quad q > 0,\\
                (\tau^2 x_0/2) e^{-\tau x_0}, & \text{\rm for}\quad q = 0.
                \end{array}\right.
    \eeq
\end{corollary}
\begin{proof}
    This simplification follows from applying Prop.~\ref{prop:parametric-derivative-1} to Eq.~\eqref{eq:weighted-repeated-double-lessthan-sum-simplification}. 
\end{proof}

We now prove the main part of Theorem 2---that we have an $O(q^3)$ estimator for FS---since
 \begin{multline}
        \mathbf{1}_{q=0} \frac{(\beta E_{z_0})^2}{8} + \mathbf{1}_{q>0} \frac{E_{z_0}}{e^{-\beta [E_{z_0}, \ldots, E_{z_q}]}} \sum_{r = 0}^q e^{-\frac{\beta}{2} [E_{z_0}, \ldots, E_{z_r}]} \Bigg( \frac{\beta^2}{8} E_{z_r} e^{-\frac{\beta}{2} [E_{z_r}, \ldots, E_{z_q}]} +  \frac{\beta^2}{8} e^{-\frac{\beta}{2} [E_{z_{r+1}}, \ldots, E_{z_q}]} - \\
         \sum_{i=r+1}^q (i - r)  e^{-\frac{\beta}{2} [E_{z_r}, \ldots, E_{z_q}, E_{z_i}]}  \Bigg) \estimates \int_0^{\beta/2} \tau \avg{ D_0(\tau) D_0 } \dtau 
    \end{multline}
\begin{proof}
    The estimator follows from applying Eq.~\eqref{eq:computing-weighted-repeated-double-lessthan-sum} to the estimator form given in Eq.~\eqref{eq:FS-estimator-orig-reorder}, taking special care to correctly translate the sum from $i, j = 0$ to $q$ on arguments $[x_0, \ldots, x_q]$ to the sum of $j, m = r$ to $q$ on arguments $[E_{z_r}, \ldots, E_{z_q}]$. The estimate $O(q^3)$ comes from the observation that, for each $r$, we must do $O(q^2)$ work.
\end{proof}

\subsection{Off-diagonal case follows from diagonal case}
\label{app-sub:off-diagonal-case-from-diagonal}

We now consider the case $H_1 \propto H - D_0$ and show it reduces to using the $H_1 \propto D_0$ estimators. Let
\begin{equation}
    \Gamma \equiv H - D_0,
\end{equation}
and define the associated ES and FS
\begin{align}
    \chi_E^{\Gamma}(\lambda; \beta) &= \int_0^\beta ( \avg{\Gamma(\tau) \Gamma} - \avg{\Gamma}^2) \dtau \\
    \chi_F^{\Gamma}(\lambda; \beta) &= \int_0^{\beta/2} \tau ( \avg{\Gamma(\tau) \Gamma} - \avg{\Gamma}^2) \dtau.
\end{align}
As with the diagonal case, we find $\chi_E^{H_1} \propto \chi_E^{\Gamma}$ and $\chi_F^{H_1} \propto \chi_F^{\Gamma}$. Since $H = D_0 + \sum_{j>0} D_j P_j$, then it follows that $\Gamma = \sum_{j > 0} D_j P_j$. One approach to estimate the ES and FS in this case is to derive estimators for $\avg{D_j P_j}$ and $\avg{(D_l P_l)(\tau) D_j P_j}$ and use linearity. While this does work, if $\Gamma$ contains $R$ terms, then it incurs an overhead of doing $R$ and $R^2$ estimations. This overhead can be avoided by using simple estimators for $H$ and $D_0$ quantities and linearity.

Firstly, $\avg{\Gamma} = \avg{H - D_0} = \avg{H} - \avg{D_0}$ by linearity. From prior work (see Sec.~VI.B in Ref.~\cite{ezzell2025advanced}), we know
\begin{equation}
    \label{app-eq:H-estimator}
    E_{z_q} + \mathbf{1}_{q\geq 1} \frac{e^{-\beta [E_{z_0}, \ldots, E_{z_{q-1}}}]}{e^{-\beta [E_{z_0}, \ldots, E_{z_q}]}} = \begin{cases}
        E_{z_q} & q = 0 \\
        E_{z_q} + \frac{e^{-\beta [E_{z_0}, \ldots, E_{z_{q-1}}}]}{e^{-\beta [E_{z_0}, \ldots, E_{z_q}]}} & q > 0,
    \end{cases} \estimates \avg{H}
\end{equation}
where $\mathbf{1}_{q \geq 1}$ is the indicator function that is 0 when $q < 1$ and 1 when $q \geq 1$. Since $E_z \estimates \avg{D_0}$, then by linearity we find,
\begin{equation}
    \mathbf{1}_{q\geq 1} \frac{e^{-\beta [E_{z_0}, \ldots, E_{z_{q-1}}}]}{e^{-\beta [E_{z_0}, \ldots, E_{z_q}]}} \estimates \avg{\Gamma}.
\end{equation}
This estimator can be used to estimate $\avg{\Gamma}^2.$

Next, we observe by linearity and cyclicity of the trace,
\begin{align}
    \avg{\Gamma(\tau) \Gamma} &= \avg{H^2} + \avg{D_0(\tau) D_0} - 2 \avg{D_0 H},
\end{align}
which has the full non-trivial $\tau$ dependence in the familiar diagonal correlator. As such, we can estimate $\avg{\Gamma(\tau) \Gamma}$ and ES and FS integrals thereof provided we can derive estimators for $\avg{H^2}$ and $\avg{D_0 H}$. By off-diagonal expansion manipulations, we know (e.g., Sec.~VI.B of Ref.~\cite{ezzell2025advanced}),
\begin{align}
    \label{eq:H2-estimator}
    E_{z_q}^2 + \frac{\mathbf{1}_{q \geq 1}(E_{z_q} + E_{z_{q-1}}) e^{-\beta [E_{z_0}, \ldots, E_{z_{q-1}}]} + \mathbf{1}_{q \geq 2}e^{-\beta [E_{z_{0}}, \ldots, E_{z_{q-2}}]} }{ e^{-\beta [E_{z_0}, \ldots, E_{z_q}]} } \estimates \avg{H^2}
\end{align}
and (see Sec.~VI.C in Ref.~\cite{ezzell2025advanced}), 
\begin{align}
    E_{z_0} \left(E_{z_0} + \mathbf{1}_{q\geq 1} \frac{e^{-\beta [E_{z_1}, \ldots, E_{z_q}}]}{e^{-\beta [E_{z_0}, \ldots, E_{z_q}]}} \right) \estimates \avg{D_0 H}.
\end{align}

Each of the static estimators, i.e., for $\avg{D_0}, \avg{H}, \avg{H^2}, and \avg{D_0 H}$, can be computed in constant time $O(1)$. Integrals of these $\tau$-independent quantities are also trivial. Hence, the entire complexity in evaluating $\avg{\Gamma(\tau) \Gamma}$ is $O(q^2)$ inherited from the complexity of computing $\avg{D_0(\tau) D_0}$. Similarly, the complexity of evaluating ES and FS is thus $O(1)$ and $O(q^3)$ via the estimators we derived in \cref{app:best-estimator-derivations}. Namely, we have now finished our proof of Theorems 1 and 2 by showing that the non-trivial parts of the ES and FS estimators for both the $H_1 \propto D_0$ and $H_1 \propto H - D_0$ are the same and hence have the same complexity. 

\section{Most general driving term estimators}
\label{app:general-driving-term}
In the most general case, $H_1$ contains both diagonal and off-diagonal terms and does not satisfy $H_1 \propto D_0$ or $H_1 \propto H - D_0$. Nevertheless, we can characterize the general form of $H_1$ for arbitrary Hamiltonians and use that to derive estimators for ES and FS in this arbitrary case. 

\subsection{\texorpdfstring{The PMR form of $H_1$}{The PMR form of H1}}
\label{app:pmr-form-of-h1}
In the study of quantum phase transitions, we assume $H(\lambda) = H_0 + \lambda H_1$ for some $[H_0, H_1] \neq 0$. Regardless of the specifics, this Hamiltonian can always be in PMR form,  $H(\lambda) = \sum_j D_j(\lambda) P_j$, where we choose $P_0 = \mathbb{1}$ by convention. Here, we show that one can always write 
\begin{equation}
    \label{eq:h1-special-form}
    H_1 = \sum_j \Lambda_j (\lambda) D_j (\lambda) P_j
\end{equation}
where $\Lambda_j (\lambda) $ are diagonal matrices and the $D_j (\lambda) P_j$ terms are directly taken from the PMR form of $H$. The consequence of this form is that there are no fundamental obstacles to estimating $\avg{H_1}, \avg{H_1(\tau) H_1},$ and integrals thereof using PMR-QMC, and in fact, estimators have a relatively simple form. A full understanding of why this is true is subtle and entrenched in PMR-QMC details, and hence, out of the scope of this work. For additional context, we point to Section~V.B in Ref.~\cite{barash2024QuantumMonteCarlo} and Ref.~\cite{ezzell2025advanced}, which discuss measurements of arbitrary operators $O$ and what properties $O$ can have in relation to $H$ that prevent accurate estimation by a single PMR-QMC simulation. To reiterate our main point, proving Eq.~\eqref{eq:h1-special-form} is one way to show that $H_1$ does not have any of these subtle problems that prevent estimation.    

To begin with, we build up the PMR form of $H(\lambda)$ by first writing PMR forms for $H_0$ and $H_1$ separately. For simplicity, denote $P \in H_0$ to mean there exists a term $DP$ in the PMR expansion of $H_0$ where $D \neq 0$. With this notation, we can write,
\begin{equation}
    H_0 = \sum_{P \in H_0} A_P P,  \  
    H_1 = \sum_{P \in H_1} B_P P  \implies H(\lambda) = \sum_{P \in H_0,  P \in H_1} (A_P + \lambda B_P) P + \sum_{P \in H_0,  P \notin H_1}  A_P P +  \sum_{P \notin H_0,  P \in H_1} \lambda B_P P.
\end{equation}
When compared with $H(\lambda) = \sum_j D_j(\lambda) P_j$, we can readily identify $\{P_j\}_j = \{ P \in H_0\} \cup \{P \in H_1 \}$ and $D_j$ as the diagonal in front of $P_j$, i.e., either $A_j, \lambda B_j, $or $(A_j + \lambda B_j)$. Before showing Eq.~\eqref{eq:h1-special-form} generally, we remark that in the extremely common cases that $H_1$ is either purely diagonal or purely off-diagonal, things simplify and Eq.~\eqref{eq:h1-special-form} is obvious. In the pure diagonal case, $H_1 = \frac{\mathbb{1}}{\lambda} D_0$, and in the pure off-diagonal case, $H_1 = \sum_{j>0} \frac{\mathbb{1}}{\lambda} D_j$.

Now for the general case, consider $P \in H_0,P \notin H_1$ where we have $D_j = A_j$ and choose $\Lambda_j = 0$. For $P \notin H_0, P \in H_1$, we have $D_j = B_j$ and choose $\Lambda_j = \mathbb{1}$.  Finally, for $P \in H_0, P \in H_1$, we have $D_j = (A_j + \lambda B_j)$ and choose $\Lambda_j = (A_j + \lambda B_j)^{+}B_j$, where $X^+$ denotes the pseudo-inverse of $X$, which prevents diving by zero. Before justifying the pseudo-inverse---which is more subtle than any discussion so far---we remark that this case is moot if $H_1$ is purely diagonal or purely off-diagonal, since such partitions have no $P$ that is contained in both $H_0$ and in $H_1$. To see why the pseudo-inverse works, suppose $\avg{z | A_j + \lambda B_j |z} = 0$ which implies either $\avg{z | A_j | z} = \avg{z | B_j | z} = 0$ or $A_j = - \lambda B_j$. In the first case, $\avg{ z | \Lambda_j B_j | z} = 0$ for any choice of $\Lambda_j$, so the pseudo-inverse works. The second-case is not relevant since it requires fine-tuning $\lambda$, and such exceptional points can be avoided in practice while still generating $\chi_F^\beta(\lambda)$ curves.

\subsection{\texorpdfstring{Estimating $\avg{H_1}$ and $\avg{H_1(\tau) H_1}$}{Estimating <H1> and <H1(tau)H1>}}

In the prior section, we showed that we can write $H_1 = \sum_j \Lambda_j D_j P_j$ for $\Lambda_j$'s diagonal matrices whenever $H = \sum_j D_j P_j$ and $P_0 = \mathds{1}$. By linearity, it is sufficient to derive estimators for each $\avg{\Lambda_j D_j P_j}$ term, for which it is useful to separate diagonal term, $\Lambda_0$ from the general off-diagonal terms $\Lambda_l D_l P_l$ for $l > 0$. In both cases, some of the present authors have derived the estimators (see Table I in \cite{ezzell2025advanced}),
\begin{align}   
    \label{app-eq:lam0-estimator}
    \Lambda_0(z) &\estimates \avg{\Lambda_0} \\ 
    \label{app-eq:laml-dl-pl-estimator}
    \delta_{P_l}^{(q)} \frac{\Lambda_l(z) e^{-\beta [E_{z_0}, \ldots, E_{z_{q-1}}]} }{e^{-\beta [E_{z_0}, \ldots, E_{z_q}]}} &\estimates \avg{\Lambda_l D_l P_l},
\end{align}
and hence, we can evaluate $\avg{H_1}$ in $O(R)$ calls to these $O(1)$ estimators. Here, $\Lambda_l(z) \equiv \avg{z | \Lambda_l | z}$ and $\delta_{P_l}^{(q)} = 1$ if the $q^{\text{th}}$ permutation in $\Siq$ is $P_l$ and is $0$ otherwise. (Recall that products of permutations $\Siq$ partially make up PMR-QMC configurations as explained in \cref{app:pmr-qmc-background}.)

Similarly, we can estimate $\avg{H_1(\tau) H_1}$ by evaluating all $R^2$ estimators of the form $\avg{A (\tau) B}$ for $A, B \in$ 
$\{\Lambda_0$, $\Lambda_1 D_1 P_1$, $\ldots$, $\Lambda_{R-1} D_{R-1} P_{R-1}\}$.
Just as the above estimator, it is best to consider the four possible ways of combining a pure diagonal contribution with a pure off-diagonal one,
\begin{align}
    \label{app-eq:diag-diag-corr-estimator}
    \frac{ \Lambda_0 (z) }{e^{-\beta[E_{z_0}, \ldots, E_{z_q}]}} \sum_{j=0}^q \Lambda_0(z_j) e^{-\tau [E_{z_j}, \ldots, E_{z_q}]} e^{-(\beta - \tau)[E_{z_0}, \ldots, E_{z_j}]} &\estimates \avg{\Lambda_0(\tau) \Lambda_0} \\
    \frac{ \delta_{P_k}^{(q)} \Lambda_k(z) }{e^{-\beta[E_{z_0}, \ldots, E_{z_{q}}]}} \sum_{j=0}^q \Lambda_0(z_j) e^{-\tau [E_{z_j}, \ldots, E_{z_{q-1}}]} e^{-(\beta - \tau)[E_{z_0}, \ldots, E_{z_j}]} &\estimates \avg{(\Lambda_k D_k P_k)(\tau)\Lambda_0} \\
    \frac{ \Lambda_0(z) }{e^{-\beta[E_{z_0}, \ldots, E_{z_q}]}} 
    \sum_{j=1}^{q-1} \delta^{(j)}_{P_l} \Lambda_l(z_j) e^{-\tau [E_{z_j}, \ldots, E_{z_{q}}]} e^{-(\beta - \tau) [E_{z_0}, \ldots, E_{z_{j-1}}]} &\estimates \avg{\Lambda_0(\tau) (\Lambda_l D_l P_l)} \\
    \label{app-eq:simple-AtauB-correlator-estimator}
    \frac{ \delta_{P_k}^{(q)} \Lambda_k(z) }{e^{-\beta [E_{z_0}, \ldots, E_{z_q}]}} \sum_{j=1}^{q-1} {\delta_{P_l}^{(j)}} {\Lambda_l(z_{j})} e^{-\tau [E_{z_{j}}, \ldots, {E_{z_{q-1}}}]}  e^{-(\beta-\tau) [E_{z_0}, \ldots, {E_{z_{j-1}}}]} &\estimates \avg{(\Lambda_k D_k P_k)(\tau) (\Lambda_l D_l P_l)},
\end{align}
where each estimator follows as different special cases of the derivation in Sec.~VIII.A in Ref.~\cite{ezzell2025advanced}. As with the $\avg{D_0(\tau) D_0}$ estimator, each of these estimators requires $O(q^2)$ operations, so the total complexity to evaluate $\avg{H_1(\tau) H_1}$ is simply $O(R^2 q^2)$. 

Although these estimators are more complicated than the pure diagonal counterparts \cref{appsub:d0-correlator}, the big picture for ES and FS is the same. Namely, the only non-trivial $\tau$ dependence is contained in the divided differences of the correlator estimator, for which we can apply the divided difference integral relations we have already derived. 

\subsection{General ES estimator}
\label{app-sub:general-es-estimator}

As in the derivation of the diagonal ES estimator, we need only consider deriving an estimator for $\int_0^\beta \avg{H_1(\tau) H_1} \dtau$ and combine it with our $\avg{H_1}$ estimators. Although $\avg{H_1(\tau) H_1}$ really consists of four different possible estimators, each of these contains an expression of the form $e^{-\tau [\ldots]} e^{-(\beta - \tau) [\ldots]}$ for which we can apply the convolution lemma, \cref{app-lem:convolution}. For example, applying this lemma to \cref{app-eq:diag-diag-corr-estimator,app-eq:simple-AtauB-correlator-estimator}, we find 
\begin{align}
    \frac{ \Lambda_0 (z) }{e^{-\beta[E_{z_0}, \ldots, E_{z_q}]}} \sum_{j=0}^q \Lambda_0(z_j) e^{-\beta[E_{z_0}, \ldots, E_{z_q}, E_{z_j}]} &\estimates \int_0^\beta \avg{\Lambda_0(\tau) \Lambda_0} \dtau \\
    \frac{ \delta_{P_k}^{(q)} \Lambda_k(z) }{e^{-\beta [E_{z_0}, \ldots, E_{z_q}]}} \sum_{j=1}^{q-1} {\delta_{P_l}^{(j)}} {\Lambda_l(z_{j})}  e^{-\beta  [E_{z_0}, \ldots, E_{z_{q-1}}]} &\estimates \int_0^\beta \avg{(\Lambda_k D_k P_k)(\tau) (\Lambda_l D_l P_l)} \dtau. 
\end{align}

The former clearly resembles $\avg{D_0(\tau) D_0}$ and is computable in $O(q^2)$ time. Unfortunately, unlike the diagonal case simplification in Theorem 1, there does not appear any way to simplify the sum into an $O(1)$ expression since $\Lambda_0(z_j)$ has no relation to $E_{z_j}$. On the other hand, the second estimator can be simplified. Namely, since the inner divided difference has no direct $j$ dependence,
\begin{equation}
    \frac{ \delta_{P_k}^{(q)} \Lambda_k(z) }{e^{-\beta [E_{z_0}, \ldots, E_{z_q}]}} \sum_{j=1}^{q-1} {\delta_{P_l}^{(j)}} {\Lambda_l(z_{j})}  e^{-\beta  [E_{z_0}, \ldots, E_{z_{q-1}}]} = \frac{ \delta_{P_k}^{(q)} \Lambda_k(z) }{e^{-\beta [E_{z_0}, \ldots, E_{z_q}]}} e^{-\beta [E_{z_0}, \ldots, E_{z_{q-1}}]} \sum_{j=1}^{q-1} {\delta_{P_l}^{(j)}} {\Lambda_l(z_{j})},
\end{equation}
then this is actually an $O(q)$ expression. In total, this means evaluating $\chi_E^{H_1}$ for completely general $H_1$ takes $O(R^2 q^2)$ time. In the special case that $H_1$ has no pure diagonal term, however, our ES estimator needs $O(R^2 q)$ effort instead. 

\subsection{General FS estimator}
\label{app-sub:general-fs-estimator}

For the FS, the same logic as in the previous section applies except we employ the more complicated integral lemma, \cref{lemm:fidsus-integral}, to evaluate $\int_0^{\beta/2} \tau \avg{H_1(\tau) H_1} \dtau$. Doing so to \cref{app-eq:diag-diag-corr-estimator,app-eq:simple-AtauB-correlator-estimator}, we find
\begin{align}
    \frac{ \Lambda_0 (z) }{e^{-\beta[E_{z_0}, \ldots, E_{z_q}]}} \sum_{j=0}^q \Lambda_0(z_j) \sum_{r=0}^j e^{-\frac{\beta}{2} [E_{z_0}, \ldots, E_{z_j}]} \sum_{m=j}^q e^{-\frac{\beta}{2} [E_{z_r}, \ldots, E_{z_q}, E_{z_j}, E_{z_m}]} &\estimates \int_0^{\beta/2} \tau  \avg{\Lambda_0(\tau) \Lambda_0} \dtau  \\
    \frac{ \delta_{P_k}^{(q)} \Lambda_k(z) }{e^{-\beta [E_{z_0}, \ldots, E_{z_q}]}} \sum_{j=1}^{q-1} {\delta_{P_l}^{(j)}} {\Lambda_l(z_{j})} \sum_{r=0}^{j-1} e^{-\frac{\beta}{2} [E_{z_0}, \ldots, E_{z_j}]} \sum_{m=j-1}^{q-1} e^{-\frac{\beta}{2} [E_{z_r}, \ldots, E_{z_{q-1}}, E_{z_m}]}  &\estimates \int_0^{\beta/2} \tau \avg{(\Lambda_k D_k P_k)(\tau) (\Lambda_l D_l P_l)} \dtau,
\end{align}
where neither can be simplified by our prior formulas or clear new insights. As such, both require $O(q^4)$ effort to compute, and the overall complexity of a general FS estimator requires $O(R^2 q^4)$ time. 

\end{document}